\documentclass[twoside]{article}

\usepackage{algorithm2e}
\usepackage{xcolor}
\usepackage{amsmath}
\usepackage[nameinlink]{cleveref}
\usepackage{amssymb}
\usepackage{mathtools}
\usepackage{IEEEtrantools}
\usepackage{bm}
\usepackage{subfig}

\SetCommentSty{mycommfont}
\RestyleAlgo{boxruled}
\LinesNumbered
\usepackage{amsthm}
\newtheorem{theorem}{Theorem}
\newtheorem{lemma}[theorem]{Lemma}

\newtheorem{standing assumption}[theorem]{Standing Assumption}
\newtheorem{corollary}[theorem]{Corollary}

\newcommand{\X}{\mathcal{X}}
\let\P\relax
\DeclareMathOperator{\P}{\mathbb{P}}
\newcommand{\Pn}{\mathbb{P}_n}
\DeclareMathOperator{\PX}{\mathcal{P}}
\newcommand{\PkX}{\mathcal{P}_k}
\newcommand{\Pcp}{\mathcal{P}_{c,p}}

\DeclareMathOperator{\R}{\mathbb{R}}
\DeclareMathOperator{\U}{\mathbb{U}}
\DeclareMathOperator{\calU}{\mathcal{U}}
\newcommand{\Hk}{\mathcal{H}_k}
\DeclareMathOperator{\PTheta}{\mathcal{P}_{\Theta}}
\DeclareMathOperator{\Ptheta}{\mathbb{P}_{\theta}}

\newcommand{\thetamapl}{\theta^*_l}
\DeclareMathOperator{\thetamap}{\theta^*}
\DeclareMathOperator{\truedist}{\P^\ast}
\newcommand{\truedistn}{\P^\ast_n}
\DeclareMathOperator{\simiid}{\overset{iid}{\sim}}

\DeclareMathOperator{\Q}{\mathbb{Q}}
\DeclareMathOperator{\E}{\mathbb{E}}
\DeclareMathOperator{\F}{\mathbb{F}}

\DeclareMathOperator{\GEM}{\operatorname{GEM}}
\DeclareMathOperator{\MMD}{\operatorname{MMD}}
\DeclareMathOperator{\Dir}{\operatorname{Dir}}
\DeclareMathOperator{\DP}{\operatorname{DP}}

\DeclareMathOperator{\nplposteriortheta}{\Pi_{\text{NPL}}}

\DeclareMathOperator{\nplposterior}{\hat{\nu}}
\DeclareMathOperator{\nplexactposterior}{\nu}

\usepackage[symbols,nogroupskip,sort=none]{glossaries-extra}

\glsxtrnewsymbol[description={True data generating distribution}]{P}{$\truedist$}
\glsxtrnewsymbol[description={Observations $x_{1:n} \simiid \truedist$}]{I}{$x_{1:n}$}
\glsxtrnewsymbol[description={Empirical measure $\Pn = \frac{1}{n} \sum_{i=1}^n \delta_{x_i}$}]{F}{$\Pn$}
\glsxtrnewsymbol[description={Centering measure of the Dirichlet Process prior $\DP(\alpha, \F)$}]{G}{$\F$}
\glsxtrnewsymbol[description={Centering measure of the Dirichlet Process posterior $\DP(\alpha', \F')$ defined as $\F'  = \frac{\alpha}{\alpha + n} \F + \frac{n}{\alpha + n} \Pn$}]{H}{$\F'$}
\glsxtrnewsymbol[description={Probability measure on $\mathcal{P}$ defined by the sampling procedure of samples from the Dirichlet approximation of the DP posterior $\DP(\alpha', \F')$} in \eqref{eq-DP-post_approx}]{A}{$\nplposterior$}
\glsxtrnewsymbol[description={Probability measure on $\mathcal{P}$ defined by the stick-break process representing samples from the exact DP posterior $\P \given x_{1:n} \sim \DP(\alpha', \F')$}]{B}{$\nplexactposterior$}
\glsxtrnewsymbol[description={Data space}]{D}{$\X$}
\glsxtrnewsymbol[description={Space of Borel distributions on $\X$}]{E}{$\PX$}
\glsxtrnewsymbol[description={Probability measure $\Ptheta \in \mathcal{P}$ indexed by parameter $\theta$ }]{C}{$\Ptheta$}
\glsxtrnewsymbol[description={Family of parametric models $\PTheta = \{\P_{\theta} : \theta \in \Theta\} \subset \PX$}]{X}{$\PTheta$}
\glsxtrnewsymbol[description={Map $\thetamapl : \mathcal{P} \rightarrow \Theta$ indexed by loss $l$ which takes as input a probability measure $\P$ and returns the value of $\theta \in \Theta$ which minimises the expected loss; i.e. $\thetamapl(\P) = \operatorname{arginf}_{\theta \in \Theta}\E_{X\sim\P}[l(X,\theta)]$
}]{K}{$\thetamapl$}
\glsxtrnewsymbol[description={Map $\thetamap : \mathcal{P} \rightarrow \Theta$ for the MMD-based loss which takes as input a probability measure $\P$ and returns the value of $\theta \in \Theta$ which minimises the MMD between $\Ptheta$ and $\P$; i.e.
$\thetamap(\P) = \operatorname{arginf}_{\theta \in \Theta}\MMD^2(\P,\Ptheta)$
}]{L}{$\thetamap$}

\newenvironment{talign*}
{\csname align*\endcsname}
{\endalign}
\newenvironment{talign}
{\align}
{\endalign}

\usepackage[style=authoryear, maxcitenames=1, uniquelist=minyear, uniquelist=false, firstinits, dashed=false, maxbibnames=9]{biblatex}
\DeclareNameAlias{sortname}{last-first}
\DeclareFieldFormat{titlecase}{\MakeSentenceCase*{#1}}
\newcommand\given[1][]{\:#1\vert\:}

\usepackage[accepted]{aistats2022}
\begin{document}

%
\runningtitle{Robust Bayesian Inference for Simulator-based Models via the MMD Posterior Bootstrap}
%

\twocolumn[

\aistatstitle{Robust Bayesian Inference for Simulator-based Models \\ via the MMD Posterior Bootstrap}

\aistatsauthor{ Charita Dellaporta \And Jeremias Knoblauch \And Theodoros Damoulas \And Fran\c{c}ois-Xavier Briol }

\aistatsaddress{ University of Warwick \And University College London \And University of Warwick \And  University College London} ]

\begin{abstract}
   Simulator-based models are models for which the likelihood is intractable but simulation of synthetic data is possible. They are often used to describe complex real-world phenomena, and as such can often be misspecified in practice. Unfortunately, existing Bayesian approaches for simulators are known to perform poorly in those cases. In this paper, we propose a novel algorithm based on the posterior bootstrap and maximum mean discrepancy estimators. This leads to a highly-parallelisable Bayesian inference algorithm with strong robustness properties. This is demonstrated through an in-depth theoretical study which includes generalisation bounds and proofs of frequentist consistency and robustness of our posterior. The approach is then assessed on a range of examples including a g-and-k distribution and a toggle-switch model.
\end{abstract}

\section{INTRODUCTION}

A key assumption of standard Bayesian inference is that the data-generating mechanism lies within the family of models defined by our choice of likelihood \parencite{bernardo2009bayesian}.
In reality, this assumption is often violated, and this can lead to inconsistent and misleading inference outcomes; see for example \textcite{grunwald2017inconsistency,owhadi2015brittleness} for the cases of linear or nonparametric regression.

The need for robust inference methods is even more apparent when dealing with simulator-based models.  In these models, independent sampling is possible, but evaluating the likelihood is intractable. Simulators are used in a range of fields including population genetics \parencite{riesselman2018deep}, ecology \parencite{beaumont2010approximate} and telecommunication engineering \parencite{bharti2021general}; see \textcite{Cranmer2020frontier} for a recent review. In each of these examples, the model is at best a rough approximation of a complex physical or biological phenomenon, and will most likely not capture all of the key characteristics of the underlying data generating process. Unfortunately, the main approaches to Bayesian inference for simulators, including approximate Bayesian computation (ABC) \parencite{beaumont2002approximate}, tend to approximate the classical Bayesian posterior, which itself lacks robustness to misspecification \parencite{Frazier2020}. This leads ABC posteriors to  exhibit poor coverage; see  \textcite{Hermans2021} for an empirical study.

One promising approach for addressing model misspecification is generalised Bayesian inference (GBI; e.g. \textcite[][]{hooker2014bayesian, ghosh2016robust, bissiri2016general,Jewson2018,Knoblauch2019}) and the closely-related approaches of quasi-Bayesian inference \parencite[e.g.][]{Chernozhukov2003}, PAC-Bayesian inference \parencite[e.g.][]{ShaweTaylor,catoniPACBound}, Safe-Bayes \parencite{Grunwald2011}, Split-BOLFI approach for higher-dimensional inference \parencite[][]{thomas2020split} and the Focused Bayesian Prediction approach \parencite[][]{loaiza2021focused}. 
These approaches typically work directly with the original model, but attempt to correct for possible misspecification by scoring observations robustly. 
Recent work has focused on GBI for intractable likelihood models \parencite{schmon2020generalized,pacchiardi2020score, thomas2020generalised, Matsubara2021}.
While some of these techniques address robustness concerns in simulator-based methods, they have two main drawbacks:
Firstly, they require setting a crucial hyperparameter that determines the relative importance of the prior for the posterior inferences.
Secondly, some
make use of Monte Carlo methods, which imposes a substantial computational burden.

Another recent approach for inference under misspecification is  Bayesian nonparametric learning (NPL) \parencite[see][]{lyddon2018nonparametric, lyddon2019general, fong2019scalable}.
Unlike GBI, NPL does not address misspecification by robustly scoring the statistical model.
Instead, robustness is achieved by obtaining a non-parametric posterior directly on the data-generating process.
This posterior then implies a  posterior on the parameter of interest through the use of a robust loss function.

Our paper's contribution is the first NPL-based algorithm for simulators. 
Specifically, we leverage the NPL framework to obtain a posterior belief distribution about the  parameters that minimise the maximum mean discrepancy (MMD) \parencite{gretton2012kernel} between our model and the  data-generating mechanism.
The MMD is a probability metric that is not only robust, but also easy to approximate through simulation---and therefore suitable for simulators. 
Further, the MMD has numerous desirable theoretical robustness and generalisation properties \parencite[see][]{briol2019statistical, cherief2019finite, cherief2020mmd}. One of the main achievements of this paper is to show that these hold for our method whenever we use a bounded kernel. Additionally, unlike ABC or GBI, our approach is computationally efficient: it is trivially parallelisable, and never requires discarding parameter samples or using inherently sequential Monte Carlo methods.

The paper is structured as follows. In \Cref{sec:background}, we recall the details of NPL and introduce MMD estimators. In \Cref{sec:methodology}, we propose our novel algorithm which combines these concepts to create a scalable approach for robust Bayesian inference with simulators. Then, \Cref{sec:theory} provides theoretical results including consistency and robustness. Finally, \Cref{sec:experiments} studies the algorithm on a range of benchmark problems for simulators including the g-and-k distribution, and a toggle-switch model describing the interaction of genes through time.

\section{BACKGROUND}\label{sec:background}

Let $\X$ denote our data space, $\PX$ the space of Borel distributions on $\X$, and $\truedist \in \PX$ the true data-generating mechanism of our data. 
In Bayesian statistics, given observations $x_{1:n} = x_1,\ldots,x_n \simiid \truedist$, one chooses a parametric model $\PTheta = \{\P_{\theta} : \theta \in \Theta\} \subset \PX$. Given $\PTheta$, the Bayesian now places a prior on the parameter $\theta$, and then conditions on $x_{1:n}$ to obtain a posterior on $\theta$. For standard Bayesian inference to be well-behaved, we have to assume that $\PTheta$ is well-specified; i.e  $\exists{\theta_0} \in \Theta$ such that $\P_{\theta_0} = \truedist$. 
When this assumption is violated, we call the model misspecified.

Misspecification in the Bayesian context has recently seen increasing interest  since the Bayesian posterior does not provide robust parameter inferences \parencite[see][]{safeBayes}.
This has led to GBI and NPL approaches aimed at rectifying the issue.
Another approach, albeit not specifically for simulators, is the BayesBag algorithm \parencite[][]{huggins2019robust}. In this paper, we combine the strengths of GBI and NPL approaches for robust inference with simulators.

\subsection{Simulator-Based Inference} \label{sec:likelihoodfree}

The problem of simulator-based models is a significant challenge for Bayesians.  Consider some $\P_\theta \in \PTheta$ with fixed $\theta \in \Theta$, whose density is the likelihood $p(\cdot \given \theta)$ and suppose we have a prior $\pi(\theta)$ on the parameter. The corresponding posterior density is given by 
\begin{talign*}
\pi(\theta \given x_{1:n}) =  \frac{ \prod_{i=1}^n p(x_i \given \theta) \pi(\theta)}{\int_{\Theta} \prod_{i=1}^n p (x_i \given \theta) \pi(\theta) d \theta} \propto  \prod_{i=1}^n p(x_i \given \theta) \pi(\theta).
\end{talign*}
The latter is intractable whenever the likelihood $p(\cdot \given \theta)$ cannot be evaluated pointwise. This has led to the development of simulator-based inference methods (sometimes also called likelihood-free inference methods).

In this paper, we focus on simulator-based models (also called generative models), which are parametric families.
This means that $\P_{\theta}$ can be represented using a distribution $\U$ on a space $\calU$ and a simulator $G_{\theta}: \calU \rightarrow \X$
so that a sample $y \sim \P_\theta$ from the model can be obtained by first sampling $u \sim \U$, and then applying the simulator $y := G_{\theta}(u) \in \X$.

\paragraph{Approximate Bayesian Computation}

ABC algorithms are arguably the most popular family of techniques for tackling Bayesian posteriors of simulator-based models
\parencite[see ][]{beaumont2002approximate, sisson2018abc, beaumont2019approximate}.
Most ABC algorithms are a variation of the following steps:
\begin{enumerate}
\itemindent=0pt
\itemsep=0em 
    \item [(i)] For $b=1,2,...B$ and prior $\pi$, sample $\theta_{b} \simiid \pi$;
    \item [(ii)] For each  $\theta_b$, sample $m$ realisations from $\P_{\theta_b}$ (i.e. simulate $u_{1:m} \simiid \U$ and set $y^{(b)}_i = G_{\theta_b}(u_i)$ $\forall i$);
    \item [(iii)] Compare $y^{(b)}_{1:m}$ with the true data $x_{1:n}$ using a discrepancy $D:\PX \times \PX \rightarrow \R$: $D(\P_n,\widehat{\P}_{\theta_b})$
    where $\P_n = \frac{1}{n}\sum_{i=1}^n \delta_{x_i}$ and $\widehat{\P}_{\theta_b} = \frac{1}{m}\sum_{j=1}^m \delta_{y^{(b)}_{j}}$.
    \item [(iv)] Weight $\theta_b$ as an approximate sample from the posterior according to $D(\P_n,\widehat{\P}_{\theta_b})$.
\end{enumerate}

The function $D$ is often chosen to be a discrepancy comparing summary statistics of the two datasets, but could also be a probability metric. For example, \textcite{bernton2019approximate} used the Wasserstein distance, whilst \textcite{park2016k2} used the MMD.
Step (iv) is usually implemented by verifying whether the discrepancy is smaller than some threshold $\varepsilon$, accepting the sample $\theta_b$ if so, and rejecting it otherwise. This leads to the following approximation of the standard Bayesian posterior density:
\begin{talign*}
    \pi_{\text{ABC}}(\theta \mid x_{1:n}) & \propto \int_{\X} \ldots \int_{\X} \prod_{j=1}^m 1_{\{D(\truedistn,(\Ptheta)_m) < \varepsilon\}}(\theta) \\
    & \qquad \qquad \times p(y_{j}|\theta)\pi(\theta) dy_{1} \ldots dy_{m}.
\end{talign*}
Smaller values of $\varepsilon$ imply an increased computational cost since more parameter samples will be rejected. However, as $\varepsilon \rightarrow 0$, the ABC posterior approaches the standard Bayesian posterior $\pi(\theta|x_{1:n})$. The latter is usually seen as a desirable property of ABC, but it can lead to poor performance because the Bayesian posterior itself lacks robustness to misspecification.

Although the prototype ABC algorithm above is parallelisable, it will tend to be inefficient because many samples will be rejected whenever $\epsilon$ is small. As a result, it is common to use inherently sequential algorithms such as population Monte Carlo \parencite{Beaumont2009} or sequential Monte Carlo with adaptive resampling \parencite{DelMoral2012}.

\paragraph{Alternative Approaches}
There are a number of perhaps less prominent alternative approaches for Bayesian inference for simulators, including Bayesian synthetic likelihoods \parencite{Price2017} and techniques relying on neural density estimation \parencite[see e.g.][]{papamakarios2016fast, papamakarios19a, Cranmer2020frontier}. Just like ABC, all these approaches are non-robust since they approximate the standard Bayesian posterior.

\subsection{Generalised Bayesian Inference (GBI)}

To address the robustness concern for standard Bayesian inference, GBI approaches have recently been proposed.
In GBI, $l_n: \X^n \times \Theta \to \R$ denotes any (empirical) loss function and $\beta>0$ a learning rate. Then, the associated GBI posterior's density is
\begin{talign*}
\pi_{\text{GBI}}(\theta \given x_{1:n}) =  \frac{  \exp\{- \beta l_n( x_{1:n},\theta)\} \pi(\theta)}{\int_{\Theta}  \exp\{- \beta l_n(x_{1:n},\theta)\} \pi(\theta) d \theta}.
\end{talign*}
Here, the loss is typically chosen in relation to the
likelihood model $p(x_{1:n} \given \theta) = \prod_{i=1}^np(x_i\given\theta)$.
Indeed, it is easy to see that for the standard Bayesian posterior, one chooses $l_n(x_{1:n}, \theta) = -\log p(x_{1:n}|\theta)$.
From this, it also becomes clear that the Bayes posterior's lack of robustness is intimately related to choosing the loss function $l_n(x_{1:n}, \theta) = -\log p(x_{1:n}|\theta)$.
To address this, a host of generalised posteriors have been derived by choosing a discrepancy $D$ with desirable robustness properties, and then seeking to find a loss so that $l_n(x_{1:n},\theta) \approx D(\truedist, \P_{\theta})$ \parencite[see][]{Jewson2018}. 

GBI for robustness in simulator-based inference has been studied in \textcite{pacchiardi2021generalized} and \textcite{schmon2020generalized}. 
The main drawback of the proposals in both papers are two-fold: firstly,  the uncertainty quantification properties of $\pi_{\text{GBI}}(\theta \given x_{1:n})$ depend on $\beta$, which can only be chosen based on rough heuristics \parencite[see][]{wu2020comparison}.
Furthermore, the methods require techniques that are computationally more cumbersome: \textcite{pacchiardi2021generalized} use an 
inherently sequential  sampling method based on pseudo-marginal MCMC \parencite{Andrieu2009}, and \textcite{schmon2020generalized} relies on standard ABC methods. 

In this paper, we perform Bayesian inference under model misspecification through a robust discrepancy like in GBI methods---but without relying on computationally burdensome methods like in ABC.

\subsection{Bayesian Nonparametric Learning (NPL)} \label{sec-NPL}
While the Bayesian nonparametric learning (NPL) framework of \textcite{lyddon2018nonparametric} was introduced to deal with misspecification, it also possesses attractive computational properties.
NPL defines a nonparametric prior directly on the true data-generating process $\truedist$, which in turn leads to a nonparametric posterior on $\truedist$.
Any posterior on the parameter space $\Theta$ of $\PTheta$ induced by NPL thus derives from this nonparametric posterior on $\PX$. 
 \begin{figure}[t!] 
  \centering
     \includegraphics[width=\linewidth]{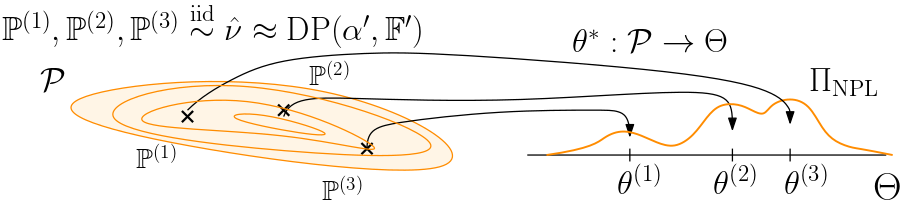}
  \caption{\emph{Sketch of the posterior bootstrap.} Samples from the NPL posterior on $\Theta$ are obtained by mapping realisations of the posterior $\nplposterior$ through the map $\thetamap$.}
  \label{fig:MMDbootstrap_illustration}
 \end{figure}
Note that this is different, but closely related to, standard Bayesian inference where conditioning occurs at the level of parameters as opposed to the level of the data-generating process.

Following  \textcite{lyddon2018nonparametric} and \textcite{fong2019scalable}, we use a Dirichlet Process (DP) prior on the data-generating process: $\P \sim \DP(\alpha, \F)$. 
Here, $\alpha > 0$ is a concentration parameter and $\F \in \PX$ a centering measure. Given $x_{1:n} \simiid \truedist$, it follows by conjugacy that
\begin{talign} 
\begin{split}
\P \given x_{1:n} & \sim \DP\left(\alpha', \F' \right), \\ \alpha' = \alpha+n,\qquad & \F' = \frac{\alpha}{\alpha + n} \F + \frac{n}{\alpha + n} \Pn \label{eq-DPPosterior}
\end{split}
\end{talign}
where $\Pn := \frac{1}{n} \sum_{i=1}^n \delta_{{x}_i}$ and $\delta_x$ is a Dirac measure at $x \in \X$. The size of $\alpha$ quantifies our confidence in the quality of the prior centering measure $\F$ and regulates the influence of the prior  on the posterior.
Accordingly, the limiting case of $\alpha =0$ corresponds to a non-informative prior and results in the posterior  $\DP(n, \Pn)$.
Hence, unlike the hyperparameter $\beta$ in GBI, the heuristic $\alpha = 0$ has a clear interpretation and yields reliable uncertainty quantification \parencite[see][]{fong2019scalable, knoblauch2020robust, galvani2021bayesian}.

The posterior on the data-generating process is readily translated into a posterior on a parameter space $\Theta$.
Suppose we know the \textit{true} data-generating process $\truedist$ and a loss $l:\X \times \Theta \rightarrow \R$.
Clearly, having access to $\truedist$ is equivalent to having access to infinitely many data points.
Consequently, no Bayesian uncertainty is needed, and one can simply compute
\begin{talign} \label{eq-thetaoff}
\thetamapl(\truedist) := \arg\inf_{\theta \in \Theta} \E_{X \sim \truedist}[l(X, \theta)].
\end{talign}
In practice, we do \textit{not} know the true data-generating process.
However, we have a posterior belief over it, so that the 
simple push-forward measure $(\thetamapl)_{\#}(\DP(\alpha', \F' ))$ gives a posterior on $\Theta$ denoted by $\nplposteriortheta$.
This is illustrated in \Cref{fig:MMDbootstrap_illustration}.

For the reader who is not familiar with pushforward distributions, $\nplposteriortheta$ is best understood through the following sampling mechanism to obtain independent realisations from $\nplposteriortheta$.
At iteration $j$, 
\\[-0.75cm]
\begin{enumerate}
    \item Sample $\P^{(j)}$ from the posterior DP in \eqref{eq-DPPosterior};
    \\[-0.65cm]
    \item Compute $\theta^{(j)} = \thetamapl(\P^{(j)})$ using \eqref{eq-thetaoff}.
    \\[-0.7cm]
\end{enumerate}
This procedure is trivially parallelisable and does not discard any samples.
It therefore overcomes the computational inefficiencies of ABC-based methods, whose rejection rate for samples is typically quite high---especially as $\varepsilon$ is moved closer to zero.

 Exact sampling from a DP as in step 1 above is usually infeasible.
 The most common approximation is the truncated stick-breaking  procedure \parencite{sethuraman1994constructive}.
 This procedure in turn can be approximated by the Dirichlet approximation of the stick breaking process \parencite{muliere1996bayesian,ishwaran2002exact}. In our case, this leads to
\begin{talign} \label{eq-DP-post_approx}
 \tilde{x}^{(j)}_{1:T} \simiid \F&, 
 \quad (w^{(j)}_{1:n}, \tilde{w}^{(j)}_{1:T}) \sim \Dir(1, \dots, 1, \frac{\alpha}{T}, \dots, \frac{\alpha}{T}). \nonumber \\
 \P^{(j)} & \approx \sum_{i=1}^n w_i^{(j)} \delta_{x_i} + \sum_{k=1}^{T} \tilde{w}_k^{(j)} \delta_{\tilde{x}_k^{(j)}}.
 \end{talign}
 $\nplposterior$ denotes the probability measure on $\PX$ defined by  \eqref{eq-DP-post_approx}, so that $\P = \sum_{i=1}^n w^{(j)} \delta_{x_i} + \sum_{k=1}^T \tilde{w}_k^{(j)} \delta_{\tilde{x}_k^{(j)}} \sim \nplposterior$. 

It is generally not possible to obtain the minimiser in \eqref{eq-thetaoff} in closed form, and this objective may not even be convex. 
This necessitates the use of numerical optimisers like stochastic gradient descent, so that step 2 above is typically only performed approximately.

In the current paper, we use the NPL framework with a loss $l$ that corresponds to the Maximum Mean Discrepancy (MMD)---a robust discrepancy popular in GBI methods  \parencite{cherief2020mmd, pacchiardi2021generalized}.
This implies that computationally, the second step in our NPL algorithm amounts to minimum distance estimation as introduced in \textcite{briol2019statistical}.

\subsection{Minimum Distance Estimation with Robust Discrepancies}\label{sec:minimumMMDestimators}

Since NPL depends on a minimisation step, we will revisit a branch of frequentist statistics whose theory and methodology we extensively draw on for our algorithm:  Minimum distance estimators (MDEs) \parencite{Parr1980}.
MDEs are a frequentist approach to parameter estimation. Given a discrepancy $D:\PX \times \PX \rightarrow \R_+$, the MDE is
\begin{talign}\label{eq:MDE}
\hat{\theta}_n := \arg\inf_{\theta \in \Theta} D(\truedistn,\P_{\theta}).
\end{talign}
MDEs can be robust to model misspecification when the underlying discrepancy is chosen with this property in mind.
A common choice of discrepancy $D$ is integral pseudo-probability metrics \parencite{muller1997integral}:
\begin{talign*}
\text{IPM} &(\P, \Q) = \sup_{f \in \mathcal{F}}  \left |\int_{\X} f(x) \P(dx) - \int_{\X} f(y) \Q(dy) \right |.
\end{talign*}
IPMs can be thought of as comparing a family of summary statistics indexed by the class $\mathcal{F}$. There are two common IPMs in the context of simulators:

\textbf{(i) Wasserstein Distance} Let $c:\X \times \X \rightarrow [0,\infty)$ be a metric and let $p \geq 1$. Furthermore, let $\mathcal{P}_{c,p} = \{ \P \in \PX \given \int_{\X} c^p(x,y) \P(dx) < \infty \; \forall y \in \X \}$. Then the Wasserstein distance is a map $W: \Pcp \times \Pcp \rightarrow \R_+$ obtained by considering an IPM with $\mathcal{F}_W := \{f:\X \rightarrow \R \given \forall x,y \in \X, |f(x)-f(y)| \leq c(x,y)\}$. The Wasserstein distance was used for MDE by \textcite{Bassetti2006,Bernton2017} and, as previously mentioned, was used for ABC by \textcite{bernton2019approximate}.

\textbf{(ii) Maximum Mean Discrepancy} Let $\Hk$ be a reproducing kernel Hilbert space (RKHS) with kernel $k: \X \times \X \rightarrow \R$ and norm $\| \cdot \|_{\Hk}$. Let $\PkX = \{ \P \in \PX \given \int_{\X} \sqrt{k(x,x)} \P(dx) < \infty \}$. Then the MMD is a map $\MMD:\PkX \times \PkX \rightarrow \R_+$ obtained by considering an IPM with $\mathcal{F}_{\text{MMD}} := \{ f:\X \rightarrow \R \given \|f\|_{\Hk} \leq 1\}$. When the kernel $k$ is characteristic \parencite{sriperumbudur2010hilbert}, the MMD is a probability metric.  A common characteristic kernel is the Gaussian kernel  $k(x, x') = \exp(-\|x-x'\|_2^2 / (2 l^2))$ where $l >0$.

The MMD was first considered for MDE by \textcite{briol2019statistical} and was further explored in \textcite{cherief2019finite,Alquier2020,Niu2021}. It is closely related to the use of MMD in generative adversarial networks \parencite{dziugaite2015training,li2015generative}, and was used for goodness-of-fit testing with composite hypotheses in \textcite{Key2021}. As previously mentioned, it has also been used for ABC  \parencite{park2016k2} and GBI \parencite{cherief2020mmd,pacchiardi2021generalized}. 

The Wasserstein distance and MMD are both popular in the context of simulators because they can both be computed or approximated for empirical measures. One can therefore use the simulator to sample $m$ realisations from $\Ptheta$, then use $D(\widehat{\P}_{\theta_b},\truedistn)$ as an approximation of $D(\Ptheta,\truedistn)$ in \eqref{eq:MDE}. In this context, the MMD has the advantage over the Wasserstein distance in that it is a robust distance \parencite{briol2019statistical,cherief2019finite}, and can be computed in quadratic, rather than cubic, time in $n$ and $m$.

\section{METHODOLOGY}\label{sec:methodology}

Our paper proposes \emph{Bayesian nonparametric learning  with the MMD} for robust and scalable inference in simulator models. 
Compared to ABC, our method has two main computational advantages: it does not discard {any} samples, and it is trivially parallelisable.
Furthermore, we will show in \Cref{sec:theory} that the approach  inherits the robustness properties of both the NPL framework and the MMD; and therefore satisfies a number of desirable properties---including finite-sample generalisation bounds, robustness guarantees, and frequentist consistency. Notably, all of these guarantees hold under model misspecification,
highlighting the approach's usefulness for simulator models of complex data generating mechanisms.

Assume we have observed data $x_{1:n} \simiid \truedist$ and are interested in inference with a parametric family $\PTheta$ of simulator-based models.
Our approach is to use the NPL framework with the loss given by the kernel scoring rule $l_k$ \parencite{Eaton1982,Dawid2007}:
\begin{talign*}
l_k(x,\theta) &= k(x,x) - 2 \int_{\X} k(x,y) \P_{\theta}(dy)\\
& \qquad + \int_{\X \times \X} k(y,z) \P_{\theta}(dy)\P_{\theta}(dz).
\end{talign*}
For this loss, the minimiser in \eqref{eq-thetaoff} becomes the MMD estimator of \textcite{briol2019statistical} given by:
\begin{talign}
\thetamap(\P) & := \arg\inf_{\theta \in \Theta} \MMD^2(\P,\P_{\theta}) \nonumber\\
\begin{split}
& = \arg\inf_{\theta \in \Theta} \int_{\X \times \X} k(x,y) \P_{\theta}(dx)\P_{\theta}(dy)  \\
 &\qquad - \frac{2}{n} \sum_{i=1}^n \int_{\X} k(x_i,x)\P_{\theta}(dx). \label{eq:thetamap}
 \end{split}
\end{talign}
This objective can easily be  approximated by sampling from both measures and using a U-statistic.
The same holds for its gradient, which naturally leads to a stochastic gradient descent algorithm. Full details are provided in \Cref{app:experiments}.

\IncMargin{1em}
\begin{algorithm}[h] 
\SetAlgoLined
\SetKwInOut{Input}{input}
 \Input{$x_{1:n}$, $T$, $B$, $\alpha$, $\F$, $\U$, $G_{\theta}$.}
 \For{$j\gets1$ \KwTo $B$}{
 Sample $\tilde{x}_{1:T}^{(j)} \simiid \F$ and \\
 $(w_{1:n}^{(j)}, \tilde{w}_{1:T}^{(j)} ) \sim \Dir\left(1, \dots, 1, \frac{\alpha}{T}, \dots, \frac{\alpha}{T}\right)$. \\
 Set $\P^{(j)} = \sum_{i=1}^n w_i^{(j)} \delta_{x_i} + \sum_{k=1}^{T} \tilde{w}_{k}^{(j)}\delta_{\tilde{x}_k^{(j)}}$. \\
 Obtain $\theta^{(j)} = \theta^\ast(\P^{(j)})$ using numerical optimisation. \\
 }
  \Return Posterior bootstrap sample $\theta^{(1:B)}$
 \caption{MMD Posterior Bootstrap}
 \label{alg:MMD_posterior_bootstrap}
\end{algorithm}
\DecMargin{1em}

Pseudo code of our approach is given in Algorithm \ref{alg:MMD_posterior_bootstrap}, and
we call the resulting procedure \textit{MMD posterior bootstrap} to pay homage to the fact that in the limiting case of $\alpha \rightarrow 0$, it is computationally similar to the Generalised Bayesian Bootstrap of \textcite{lyddon2019general}. Incorporating prior information in the Posterior Bootstrap has also recently been discussed in \textcite{pompe2021introducing}.
From a practical standpoint, this limiting case is particularly interesting because it  eliminates  dependence on the centering measure $\F$. 
This corresponds to a non-informative prior  \parencite[see][]{fong2019scalable, knoblauch2020robust}, and reflects our uncertainty about the model in a misspecified setting.

\section{THEORY}\label{sec:theory}

Before presenting our experiments, we provide a theoretical study for which we impose:
\begin{standing assumption} \label{as1}
$| k(x,y) | \leq 1 \quad \forall x,y \in \X$.
\end{standing assumption}
While the kernel is required to be bounded, the choice of upper bound $1$ is without loss of generality, since we only consider the minimiser of the MMD, which does not depend on the bound's magnitude.
Note that the condition also ensures that $\F$, $\truedist$, and any element of $\PTheta$ are in $\PkX$.
While no other assumption on the kernel is needed for our theory, it is desirable for the kernel to be characteristic---as this guarantees that the MMD is a metric on $\PkX$ and hence that $\truedist$ can be recovered in well-specified models. 

Since the conjugacy of the DP posterior in the NPL setting relies on the assumption that the data are i.i.d, we also inherit this requirement:
\begin{standing assumption} \label{as2}
 $x_{1:n} \simiid \truedist$.
\end{standing assumption}
Our results provide the first generalisation, robustness, and consistency guarantees for NPL posteriors.
Our first result is a generalisation bound in terms of the MMD expected under $\nplposterior$ (\Cref{gen_error_bound}). 
Beyond that, we also prove that consistency in the frequentist sense (\Cref{post_consistency_miss}), and robustness to outliers (\Cref{cor-contam-mis}). 
A particular characteristic is that all of these results hold for the misspecified setting where $\truedist \not\in \PTheta$. The well-specified counterparts are Corollaries \ref{cor-gen-error-well}, \ref{post_consistency_well}, and \ref{cor-contam-well} in Appendix \ref{app:proofs}, and are obtained by noting that in this case  $\inf_{\theta \in \Theta} \MMD(\Ptheta, \truedist) = 0$. 

\subsection{Generalisation Error}
First, we bound the  generalisation error of our procedure, i.e. the error we expect based on unseen data from the true data-generating mechanism.
Here, the notion of error considered is the MMD to be expected under the (approximated) posterior, which is given by $\E_{\P \sim \nplposterior} \left[\MMD(\truedist, \P_{\theta^\ast(\P)})\right]$. We therefore bound the expected value of this quantity under unseen data from the data-generating mechanism:
\begin{theorem} \footnote{An earlier version of the paper contained a small error in this bound. The overall rate of the bound and discussions hold the same.}\label{gen_error_bound}
\begin{talign*}
    &\E_{x_{1:n} \simiid \truedist} \left[ \E_{\P \sim \nplposterior} \left[\MMD(\truedist, \P_{\theta^\ast(\P)}) \right] \right] \\
    &\leq \inf_{\theta \in \Theta}  \MMD( \truedist,\P_{\theta}) 
    + \frac{2}{\sqrt{n}} 
    + 2\sqrt{\frac{2(n - 1) + \alpha(\alpha+1)}{(\alpha+n) (n + \alpha + 1)}} \\
    &\quad \quad+ 2\sqrt{\frac{\alpha (1 + \alpha)}{(\alpha +n)(\alpha+n+1)}}.
\end{talign*}
\end{theorem}
Since the expectation over $\E_{\P \sim \nplposterior} \left[\MMD(\truedist, \P_{\theta^\ast(\P)})\right]$ is lower-bounded by $\inf_{\theta \in \Theta}  \MMD( \P_{\theta}, \truedist)$, 
this result tells us that, in expectation, the error is at most ${2}/{\sqrt{n}} + {2\sqrt{2(n-1)+\alpha(\alpha+1)}/{\sqrt{(\alpha + n)(\alpha + n + 1)}}}
    + {2 \sqrt{\alpha(1 + \alpha)}}/{\sqrt{(\alpha + n)(\alpha + n + 1)}}$. 
All terms vanish as $n\to\infty$, which implies a $\sqrt{n}$-rate as the first two terms dominate the third. This is the same rate as for the frequentist MMD estimators of \textcite{briol2019statistical} (Theorem 1), as well as the MMD-based GBI of \textcite{cherief2019finite} (Theorem 3.1).

 \begin{figure*}[t!]
\centering
\includegraphics[width=0.99\textwidth]{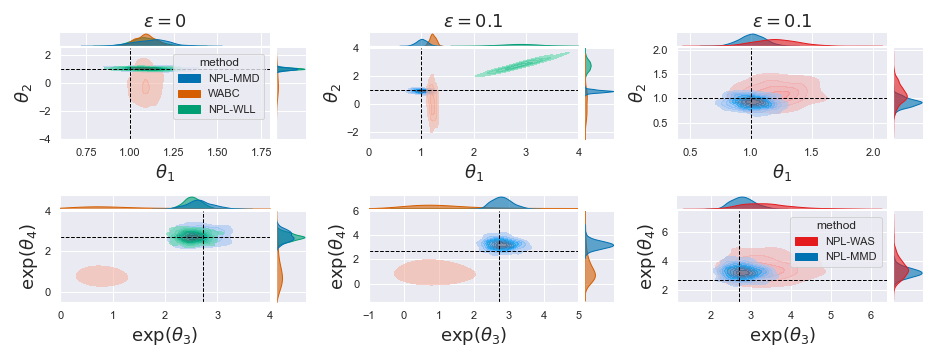}
\vspace{-5mm}
\caption{Posterior marginal distributions for the Gaussian location model in $d=4$. The true parameter $\theta_0$ is indicated by dotted lines. (Left \& middle) Comparison against WABC and NPL-WLL methods for the well-specified case and $\epsilon=0.1$. (Right) Comparison against NPL with the Wasserstein distance for $\epsilon = 0.1$. We note that in the low middle panel, the NPL-WLL method is not visible as the samples lie significantly away from $\theta_0$. }
\label{fig-post-marg-gauss}
\end{figure*}
 \begin{figure*}[t!]
\centering
\includegraphics[width=0.75\textwidth]{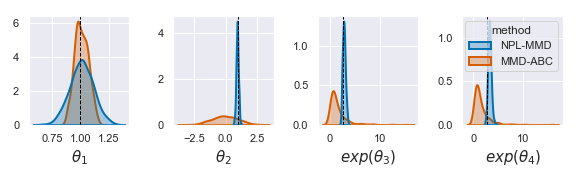}
\caption{Marginal densities of the posterior obtained by NPL-MMD and ABC with the MMD for the Gaussian location model in $d=4$ for $\epsilon = 0.1$. The true parameter $\theta_0$ is indicated by dotted lines. }
\label{fig-mabc}
\end{figure*}
\subsection{Posterior Consistency}
To deepen our analysis, we consider consistency in the frequentist sense.
This guarantees that the posterior contracts around the optimal value $\theta^{\ast}(\truedist)$
as $n\to\infty$. 
For standard Bayesian inference with posterior measure $\Pi_n$ (whose density was previously denoted $\pi(\cdot \given x_{1:n})$) defined on $\Theta$ directly,
and where our model is misspecified so that $\exists C > 0$ such that  $\inf_{\theta \in \Theta} \MMD(\P_{\theta}, \truedist) = C$,
this would amount to 
\begin{talign} \label{Pi-case-consistency}
    \Pi_n\left( 
            \theta \in \Theta:
            \MMD(\P_{\theta}, \truedist) > 
         C+ \frac{M_n}{n^{1/2}}
    \right)
     \overset{n\to\infty}{\longrightarrow}  0
\end{talign}
for a sequence $M_n \rightarrow + \infty$ so that $M_nn^{-\frac{1}{2}} \rightarrow 0$ as $n \rightarrow \infty$.
In other words, the posterior measure over regions of $\Theta$ that induces large values for $\MMD(\P_{\theta}, \truedist)$ goes to $C$ as we obtain more data so that it must ultimately concentrate around increasingly small neighbourhoods of  optimal MMD-minimising values for $\theta$. 

In our case, the posterior $\nplposteriortheta$ is defined implicitly: given the function $\theta^\ast(\P)$ and
the approximate posterior $\nplposterior$ on $\P$ constructed via equation \eqref{eq-DP-post_approx}, our posterior on $\theta$ is obtained by the push-forward operation.
Thus, the equivalent statement for our case concerns $\nplposterior$:
\begin{theorem} \label{post_consistency_miss}
Suppose our model is misspecified so that for some $C > 0$ we have  $\inf_{\theta \in \Theta} \MMD(\P_{\theta}, \truedist) = C$. Then, we have that for any $M_n \rightarrow + \infty$:
\begin{talign*}
    \nplposterior \left(  \P \in \PX :
    \MMD(\P_{\theta^\ast(\P)}, \truedist) > C + \frac{M_n}{n^{1/2}} \right) \overset{n\to\infty}{\longrightarrow} 0.
\end{talign*}
\end{theorem}
\subsection{Robustness to Outliers}

To assess robustness against the presence of outliers in the dataset, we consider Huber's contamination model \parencite{huber1992robust}.
In this setting, a proportion $1-\varepsilon$ of the observed data is  generated from the distribution of interest $\tilde{\P}$, 
and the rest follows a contaminating noise distribution; i.e.
$
\truedist = (1 - \epsilon) \tilde{\P} + \epsilon \Q
$
for $\tilde{\P}, \Q \in \PX$ and $\epsilon \in [0,1]$. 
Here, $\Q$ is the contaminant, and so the goal is to place most posterior mass on values of $\theta$ for which $\P_{\theta}  \approx \tilde{\P}$, where closeness is measured via the MMD.
\begin{corollary}\label{cor-contam-mis} Suppose $\truedist = (1 - \epsilon) \tilde{\P} + \epsilon \Q$. Then
\begin{talign*}
     &\E_{x_{1:n} \simiid \truedist} \big[ \E_{\P \sim \nplposterior} \big[ \MMD \big( \tilde{\P}, \P_{\theta^\ast(\P)}\big) \big] \big]  \\
     &\leq \inf_{\theta \in \Theta} \MMD(\tilde{\P} , \P_{\theta})
     + 4 \epsilon + \frac{2}{\sqrt{n}} 
    + 2\sqrt{\frac{2(n - 1) + \alpha(\alpha+1)}{(\alpha+n) (n + \alpha + 1)}} \\
    &\quad \quad+ 2\sqrt{\frac{\alpha (1 + \alpha)}{(\alpha +n)(\alpha+n+1)}}. 
\end{talign*}
\end{corollary}
Similarly to the generalisation bound discussed above, the rate at which this bound goes to zero is  $\max\{n^{-\frac{1}{2}}, \epsilon\}$. 
Since there are at most $\lceil \epsilon n \rceil$ contaminated data points in a dataset, the maximum number of outliers the dataset can have while maintaining the same rate is of order $\sqrt{n}$, which agrees with \textcite{cherief2019finite}, Corollary 3.4, who studied the frequentist minimum MMD estimator.
\section{EXPERIMENTS}\label{sec:experiments}
We now study the performance of our method using three examples. Throughout, we use the Gaussian kernel, which satisfies Standing Assumption \ref{as1}. We also use a non-informative prior by setting $\alpha = 0$. Further  experimental details and results are reported in \Cref{app:experiments}. \Cref{app:additional-exps} provides additional experiments which examine sensitivity to hyperparameters and provide comparison of our method with the MMD-Bayes method in \textcite{pacchiardi2021generalized}. We further consider an example of misspecification which is not based in a contamination model by wrongly fitting a Gaussian location model to Cauchy distributed data. The code for all experiments can be found at \url{https://github.com/haritadell/npl_mmd_project.git}. 

\subsection{Gaussian Location Model}
We start by considering a toy example, the Gaussian location model. While the likelihood of this model is available, we treat it as a simulator to study the properties of our proposed method.
We take $\P_{\theta} = \mathcal{N}(\theta, I_{d \times d})$ and use a true data generating process $\truedist = (1- \epsilon) \P_{\theta_0} + \epsilon \P_{\theta'}$, with $\theta_0 = (1,\dots, 1) \in \R^d$ and $\theta' = (20, \dots, 20) \in \R^d$ in $d= 4$. We assess robustness by considering both the well-specified case $\epsilon = 0$ and the case $\epsilon = 0.1$ and $n = 200$ realisations from $\truedist$.

Our simulation study will illustrate how robustness is inherited from both the NPL framework, which is more robust to model misspecification than standard Bayes or ABC, and the MMD being more robust than alternative losses such as the Wasserstein distance or negative log-likelihood.
\paragraph{Wasserstein ABC}
We first compare our method against the Wasserstein ABC (WABC) \parencite{bernton2019approximate} algorithm, which uses Sequential Monte Carlo (SMC). 
We compare against WABC since it is a popular algorithm which, unlike standard ABC, does not require  hand-crafting of summary statistics. 
As observed in \Cref{fig-post-marg-gauss}, NPL-MMD outperforms WABC in both the well-specified and misspecified settings. The former is explained by the fact that the Wasserstein distance exhibits poor sample complexity for $d>1$ \parencite{Fournier2015}, whereas the MMD can be estimated at a $\sqrt{n}-$rate. The latter is explained by the fact that the WABC is an approximation of the exact Bayesian posterior, which is not robust.

\paragraph{NPL with  Wasserstein distance } \label{app:npl-was}
Secondly, we consider the NPL framework using the Wasserstein distance (NPL-WAS) such that for $W:\P^2 \to \R_{\geq 0}$ the 2-Wasserstein distance,
$
\theta^{*}_{W}(\P) := \arg\inf_{\theta \in \Theta} W(\P, \P_{\theta}).
$
We use the \texttt{POT} package \parencite{flamary2021pot} for approximating the Wasserstein distance and the \textit{Powell} optimiser from \texttt{Scipy} \parencite{virtanen2020scipy}. 
We focus specifically on the misspecified setting, and notice that NPL-MMD outperforms the NPL-WAS in that case. This clearly demonstrates the advantage of using a robust loss function, even when using a robust inference framework such as NPL.

\begin{figure*}[t!]
\centering
\includegraphics[width=0.99\textwidth]{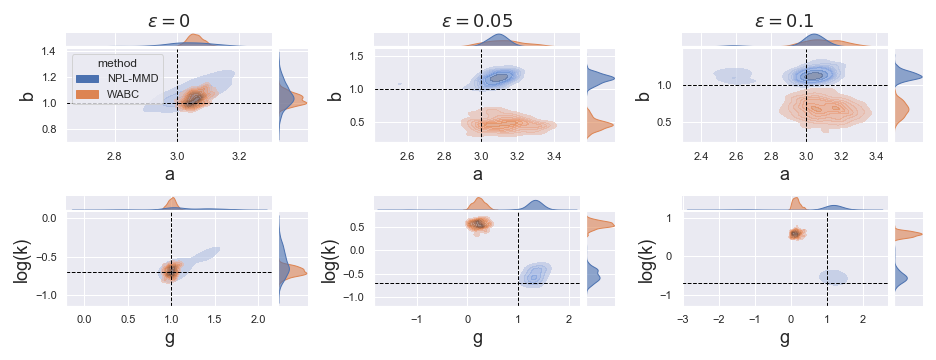}
\vspace{-5mm}
\caption{Posterior marginal distributions for $\theta$ in the univariate g-and-k model for an increasing percentage of outliers present in the dataset. The true parameter $\theta_0$ is indicated by the dotted lines.}
\label{fig-post-marg-gnk}
\end{figure*}
\paragraph{NPL with log-likelihood}
Thirdly, the availability of the likelihood in this toy problem allows for comparison with the original NPL posterior bootstrap using $l(x,\theta) = \frac{1}{\sqrt{2 \pi}} \exp\{- \frac{1}{2} (x - \theta)^2\}$ as in \textcite{lyddon2018nonparametric}, which we call the weighted log-likelihood NPL (NPL-WLL). 
\Cref{fig-post-marg-gauss} shows that NPL-WLL and NPL-MMD perform similarly in the well-specified case, but NPL-MMD significantly outperforms NPL-WLL in the misspecified case. This is once again due to the fact that the MMD is a robust loss, whereas the negative log-likelihood is not.

\paragraph{ABC with the MMD}
Finally, we further explore ABC method with the SMC samplers considered in \textcite{bernton2019approximate}, this time using the MMD instead of the Wasserstein distance. The ABC with the MMD has also been previously explored in \textcite[][]{park2016k2}. The posterior marginals are visualised in Figure \ref{fig-mabc} for $\epsilon = 0.1$. This figure clearly demonstrates the advantage of the NPL framework over ABC.
\subsection{Simulator-based Models}
We now consider two more complex numerical examples, for which likelihood-based inference is impossible. We compare solely against the Wasserstein ABC for simplicity by directly using the experimental setup in \textcite{bernton2019approximate}. 
 \begin{figure*}[t!]
\centering
\includegraphics[width=\textwidth]{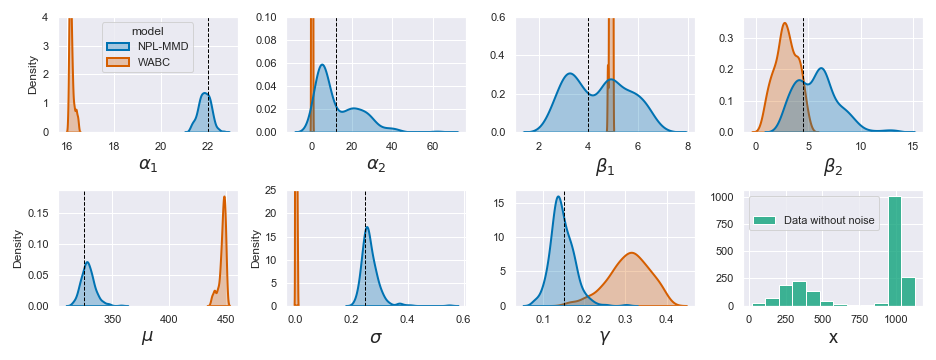}
\vspace{-5mm}
\caption{Posterior marginal distributions for $\theta$ in the toggle switch model with $10\%$ of noisy data in the dataset. Bottom right figure presents 2000 samples from the well-specified toggle switch model.}
\label{fig-post-marg-togswitch}
\end{figure*}
\paragraph{G-and-k Distribution Model} \label{sec-gandk}
First, consider $G_{\theta} : [0,1]^2 \rightarrow \R$ with  $\theta = (a, b, g, k)$ where
\begin{talign*}
G_{\theta}(u) &= a + b \left(1 + 0.8 \frac{1 - \exp(-g z(u))}{1 + \exp(-g z(u))} \right)(1 + z(u)^2)^k z(u),
\end{talign*}
$z(u) = \sqrt{-2 \log(u_1)} \cos(2 \pi u_2)$ and $\U = \text{Unif}([0,1]^2)$. Parameters $a,b,g$ and $k$ control the location, scale, skewness and kurtosis respectively. For computational convenience, we reparametrise the last parameter by setting $k' = \exp(k)$. Although $\P_\theta$ is one-dimensional, it is a popular baseline for simulator-based models \parencite{prangle2020gk} because of the challenge of inferring the four parameters simultaneously. It has also been used extensively in applications; for example to model the price of AirBnB rentals \parencite{Rodrigues2020}, the air pollution \parencite{Rayner2002} or for non-life insurance modelling \parencite{Peters2016}. 

Our data consists of $n = 2^{11}$ realisations from $\truedist = (1 - \epsilon) \P_{\theta_0} + \epsilon \Q$ where $\P_{\theta_0}$ denotes the g-and-k with $\theta_0 = (3,1,1,-\log(2))$, and  $\Q$ is the shifted distribution $\Q = \P_{\theta_0} \pm 50$ with an equal number of points shifted to either direction. The resulting posteriors are shown in Figure \ref{fig-post-marg-gnk} displayed as bivariate plots. The WABC method appears more sensitive to contamination in the dataset---in contrast to the NPL-MMD, which concentrates significantly closer to the true parameter values for an increasing proportion of outliers. This is particularly the case for the last two parameters $g$ and $\log(k)$ which are well-known to be more challenging to estimate. 
\paragraph{Toggle Switch Model}

Finally, we consider the toggle-switch model arising in Systems Biology  \parencite{bonassi2011bayesian,bonassi2015sequential}. This is a dynamic model used to study cellular networks; more precisely, the network describes the interaction of two genes $u$ and $v$ over time. The simulator is too complex to include in the main text, but is given in \Cref{app:experiments}. The data is one-dimensional, but the model has $7$ parameters and the latent space is $601$-dimensional. 

We consider inference on $\theta$ for $n=2000$ data points simulated from the toggle-switch model with true parameter $\theta_0 = (22,12,4,4.5,325,0.25,0.15)$ in which $10\%$ of the data have some added Cauchy noise of location parameter 0 and scale parameter 10. Such noise can be interpreted as measurement error in the collection of data. The posterior marginal distributions for $\theta$ are shown in Figure \ref{fig-post-marg-togswitch}, and indicate that the NPL-MMD method is successful in concentrating around $\theta_0$ despite the Cauchy noise.

\subsection{Computational Time} 
We provide the computational time recorded for each experiment over multiple independent runs. The Gaussian location and G-and-k experiments were repeated 10 times and the Toggle-switch was repeated 5 times due to the higher computational cost. For each run, a new dataset was generated and $B$ posterior samples were obtained for $\theta$. The average time for the generation of $B$ samples is recorded for each method. Our method is compatible with the use of GPU so in Table \ref{tab:times} we provide the average time recorded when run in \texttt{Google Colaboratory Pro}.
\begin{table}[h] 
\caption{Average Clock Time In Seconds.}
\begin{center}
\begin{tabular}{l||l|l|l}
      \textbf{Method}  & {\small \textbf{Gaussian}} & {\small \textbf{G-and-k}} & {\small \textbf{Toggle-S}} \\
      \hline \hline
     {\tiny NPL-MMD} & \textbf{$1.09 \times 10^2$} & $3.02 \times 10$ & $7.42 \times 10^3$  \\ 
     \hline
     {\tiny WABC} & $3.34 \times 10^3$ & $1.05 \times 10^2$ & $3.08 \times 10^4$ \\
    \end{tabular}
    \end{center}
\label{tab:times}
\end{table}
\vspace{-0.6cm}
\section{CONCLUSION}

Our paper proposes the very first posterior bootstrap method in the  simulator setting.
Unlike ABC, our method does not discard any samples and can run in parallel---leading to a substantially decreased computational burden.
Further, the approach is based on MMD estimators, and therefore inherits both the robustness of the NPL framework and that of the MMD. 
We support this claim with three theoretical results that hold even in the presence of misspecification.
The results include a generalisation bound, a robustness guarantee, and a consistency result.
We verified the practical utility of our theory through deploying it on three models, which highlighted both robustness and computational advantages of our method.

A particular strength of the method is that no assumption is required on $\X$. Hence, the results apply directly to any set $\X$ on which one can define kernels (e.g. graphs, strings, or discrete spaces).

In future work, we would like to tackle the challenges posed by the optimisation step. 
Specifically, since the MMD-objective is usually non-convex, a particularly interesting question would be if the  kernel can be used to enforce a more well-behaved objective.
Another direction for  improvement is the cost of computing the MMD minimiser.
While naively, this scales quadratically in the number of observations, it could be reduced to linear time using approaches like in Lemma 14 of \textcite{gretton2012kernel}.

\subsubsection*{Acknowledgements}
CD was funded by EPSRC grant [EP/T51794X/1] as part of the Warwick CDT in Mathematics and Statistics. 
JK was funded through the Biometrika Fellowship courtesy of the Biometrika Trust. TD acknowledges support from a UKRI Turing AI Fellowship [EP/V02678X/1]. TD and FXB were supported by the Lloyd’s Register Foundation Programme on Data-Centric Engineering and The Alan Turing Institute under the EPSRC grant [EP/N510129/1]. The authors thank all reviewers for their useful comments. 

\subsubsection*{References}

\printbibliography[heading=none]

@article{grunwald2017inconsistency,
  title={Inconsistency of Bayesian inference for misspecified linear models, and a proposal for repairing it},
  author={Gr{\"u}nwald, P and Van Ommen, Thijs},
  journal={Bayesian Analysis},
  volume={12},
  number={4},
  pages={1069--1103},
  year={2017},
  publisher={International Society for Bayesian Analysis}
}

@article{Eaton1982,
author = {Eaton, M. L.},
journal = {Statistical Decision Theory and Related Topics III},
pages = {329--352},
title = {{A method for evaluating improper prior distributions}},
year = {1982}
}

@article{Dawid2007,
author = {Dawid, A. P.},
journal = {Annals of the Institute of Statistical Mathematics},
number = {1},
pages = {77--93},
title = {{The geometry of proper scoring rules}},
volume = {59},
year = {2007}
}

@article{Niu2021,
author = {Niu, Z. and Meier, J. and Briol, F-X},
journal = {arXiv:2106.11561},
title = {{Discrepancy-based inference for intractable generative models using quasi-Monte Carlo}},
year = {2021}
}

@article{owhadi2015brittleness,
  title={{Brittleness of Bayesian inference under finite information in a continuous world}},
  author={Owhadi, Houman and Scovel, Clint and Sullivan, Tim},
  journal={Electronic Journal of Statistics},
  volume={9},
  number={1},
  pages={1--79},
  year={2015},
  publisher={Institute of Mathematical Statistics and Bernoulli Society}
}

@article{beaumont2019approximate,
  title={{Approximate Bayesian computation}},
  author={Beaumont, Mark A},
  journal={Annual Review of Statistics and Its Application},
  volume={6},
  pages={379--403},
  year={2019}
}

@article{sisson2018abc,
  title={ABC samplers},
  author={Sisson, SA and Fan, Y},
  journal={Handbook of Approximate Bayesian Computation},
  pages={87--123},
  year={2018},
  publisher={Chapman and Hall/CRC}
}

@article{Jewson2018,
author = {Jewson, J. and Smith, J. Q. and Holmes, C.},
journal = {Entropy},
number = {6},
pages = {442},
title = {{Principled Bayesian minimum divergence inference}},
volume = {20},
year = {2018}
}

@article{Matsubara2021,
author = {Matsubara, T. and Knoblauch, J. and Briol, F.-X. and Oates, C. J.},
journal = {arXiv:2104.07359},
title = {{Robust generalised Bayesian inference for intractable likelihoods}},
year = {2021}
}

@article{Knoblauch2019,
author = {Knoblauch, J. and Jewson, J. and Damoulas, T.},
journal = {arXiv:1904.02063},
title = {{Generalized variational inference: three arguments for deriving new posteriors}},
year = {2019}
}

@article{Chernozhukov2003,
author = {Chernozhukov, V. and Hong, H.},
journal = {Journal of Econometrics},
number = {2},
pages = {293--346},
title = {{An MCMC approach to classical estimation}},
volume = {115},
year = {2003}
}

@article{Grunwald2011,
author = {Gr{\"{u}}nwald, Peter},
journal = {Journal of Machine Learning Research},
pages = {397--419},
title = {{Safe learning: bridging the gap between Bayes, MDL and statistical learning theory via empirical convexity}},
volume = {19},
year = {2011}
}

@inproceedings{fong2019scalable,
  title={Scalable nonparametric sampling from multimodal posteriors with the posterior bootstrap},
  author={Fong, Edwin and Lyddon, Simon and Holmes, Chris},
  booktitle={International Conference on Machine Learning},
  pages={1952--1962},
  year={2019},
  organization={PMLR}
}

@inproceedings{lyddon2018nonparametric,
  title={{Nonparametric learning from Bayesian models with randomized objective functions}},
  author={Lyddon, Simon and Walker, Stephen and Holmes, Chris},
  booktitle={Proceedings of the 32nd International Conference on Neural Information Processing Systems},
  pages={2075--2085},
  year={2018}
}

@article{Frazier2020,
author = {Frazier, D. T. and Robert, C. P. and Rousseau, J.},
journal = {Journal of the Royal Statistical Society B: Statistical Methodology},
number = {2},
pages = {421--444},
title = {{Model misspecification in ABC: consequences and diagnostics}},
volume = {82},
year = {2020}
}

@article{DelMoral2012,
author = {{Del Moral}, Pierre and Doucet, Arnaud and Jasra, Ajay},
journal = {Bernoulli},
number = {1},
pages = {252--278},
title = {{On adaptive resampling strategies for sequential Monte Carlo methods}},
volume = {18},
year = {2012}
}

@inproceedings{Key2021,
author = {Key, O. and Fernandez, T. and Gretton, A. and Briol, F-X.},
booktitle = {arXiv:2111.10275},
title = {{Composite goodness-of-fit tests with kernels}},
year = {2021}
}

@article{Hermans2021,
author = {Hermans, J. and Delaunoy, A. and Rozet, F. and Wehenkel, A. and Louppe, G.},
journal = {arXiv:2110.06581},
title = {{Averting a crisis in simulation-based inference}},
year = {2021}
}

@article{Beaumont2009,
author = {Beaumont, M. A. and Cornuet, J-M. and Marin, J-M. and Robert, C. P.},
issn = {14643510},
journal = {Biometrika},
number = {4},
pages = {983--990},
title = {{Adaptive approximate Bayesian computation}},
volume = {96},
year = {2009}
}

@article{cherief2019finite,
  title={Finite sample properties of parametric MMD estimation: robustness to misspecification and dependence},
  author={Ch{\'e}rief-Abdellatif, Badr-Eddine and Alquier, Pierre},
  journal={arXiv preprint arXiv:1912.05737},
  year={2019}
}

@article{Alquier2020,
author = {Alquier, P. and Gerber, M.},
journal = {arXiv:2006.00840},
number = {1},
title = {{Universal robust regression via maximum mean discrepancy}},
year = {2020}
}

@article{Cranmer2020frontier,
author = {Cranmer, K. and Brehmer, J. and Louppe, G.},
journal = {Proceedings of the National Academy of Sciences of the United States of America},
number = {48},
title = {{The frontier of simulation-based inference}},
volume = {117},
year = {2020}
}

@article{pacchiardi2020score,
  title={Score matched conditional exponential families for likelihood-free inference},
  author={Pacchiardi, Lorenzo and Dutta, Ritabrata},
  journal={arXiv:2012.10903},
  year={2020}
}

@article{beaumont2002approximate,
  title={{Approximate Bayesian computation in population genetics}},
  author={Beaumont, Mark A and Zhang, Wenyang and Balding, David J},
  journal={Genetics},
  volume={162},
  number={4},
  pages={2025--2035},
  year={2002},
  publisher={Oxford University Press}
}

@article{beaumont2010approximate,
  title={{Approximate Bayesian computation in evolution and ecology}},
  author={Beaumont, Mark A},
  journal={Annual Review of Ecology, Evolution, and Systematics},
  volume={41},
  pages={379--406},
  year={2010},
}

@article{bharti2021general,
  title={A general method for calibrating stochastic radio channel models with kernels},
  author={Bharti, Ayush and Briol, Francois-Xavier and Pedersen, Troels},
  journal={IEEE Transactions on Antennas and Propagation},
  year={2021},
  publisher={IEEE}
}

@article{Bernton2017,
author = {Bernton, E. and Jacob, P. E. and Gerber, M. and Robert, C. P.},
journal = {Information and Inference: A Journal of the IMA},
number = {4},
pages = {657--676},
title = {{Inference in generative models using the Wasserstein distance}},
volume = {8},
year = {2017}
}

@article{Peters2016,
author = {Peters, G. and Chen, W. and Gerlach, R.},
journal = {Risks},
number = {2},
pages = {14},
title = {{Estimating quantile families of loss distributions for non-life insurance modelling via L-moments}},
volume = {4},
year = {2016}
}

@article{Rayner2002,
author = {Rayner, G. D. and MacGillivray, H. L.},
journal = {Statistics and Computing},
number = {1},
pages = {57--75},
title = {{Numerical maximum likelihood estimation for the g-and-k and generalized g-and-h distributions}},
volume = {12},
year = {2002}
}

@article{Rodrigues2020,
author = {Rodrigues, G. S. and Nott, David J. and Sisson, S. A.},
journal = {Statistics and Computing},
number = {4},
pages = {1057--1073},
title = {{Likelihood-free approximate Gibbs sampling}},
volume = {30},
year = {2020}
}

@article{Bassetti2006,
author = {Bassetti, F. and Bodini, A. and Regazzini, E.},
journal = {Statistics and Probability Letters},
pages = {1298--1302},
title = {{On minimum Kantorovich distance estimators}},
volume = {76},
year = {2006}
}

@article{bernton2019approximate,
  title={{Approximate Bayesian computation with the Wasserstein distance}},
  author={Bernton, Espen and Jacob, Pierre E and Gerber, Mathieu and Robert, Christian P},
  journal={Journal of the Royal Statistical Society: Series B (Statistical Methodology)},
  volume={81},
  number={2},
  pages={235--269},
  year={2019},
  publisher={Wiley Online Library}
}

@inproceedings{park2016k2,
  title={{K2-ABC: approximate Bayesian computation with kernel embeddings}},
  author={Park, Mijung and Jitkrittum, Wittawat and Sejdinovic, Dino},
  booktitle={Artificial Intelligence and Statistics},
  pages={398--407},
  year={2016},
  organization={PMLR}
}

@article{briol2019statistical,
  title={Statistical inference for generative models with maximum mean discrepancy},
  author={Briol, Francois-Xavier and Barp, Alessandro and Duncan, Andrew B and Girolami, Mark},
  journal={arXiv:1906.05944},
  year={2019}
}

@inproceedings{schmon2020generalized,
  title={{Generalized Posteriors in approximate Bayesian computation}},
  author={Schmon, Sebastian M and Cannon, Patrick W and Knoblauch, Jeremias},
  booktitle={Third Symposium on Advances in Approximate Bayesian Inference},
  year={2020}
}

@article{bissiri2016general,
  title={A general framework for updating belief distributions},
  author={Bissiri, Pier Giovanni and Holmes, Chris C and Walker, Stephen G},
  journal={Journal of the Royal Statistical Society. Series B, Statistical methodology},
  volume={78},
  number={5},
  pages={1103},
  year={2016},
  publisher={Wiley-Blackwell}
}

@article{hooker2014bayesian,
  title={{Bayesian model robustness via disparities}},
  author={Hooker, Giles and Vidyashankar, Anand N},
  journal={Test},
  volume={23},
  number={3},
  pages={556--584},
  year={2014},
  publisher={Springer}
}

@article{wu2020comparison,
  title={{A comparison of learning rate selection methods in generalized Bayesian inference}},
  author={Wu, Pei-Shien and Martin, Ryan},
  journal={arXiv:2012.11349},
  year={2020}
}

@article{ghosh2016robust,
  title={Robust Bayes estimation using the density power divergence},
  author={Ghosh, Abhik and Basu, Ayanendranath},
  journal={Annals of the Institute of Statistical Mathematics},
  volume={68},
  number={2},
  pages={413--437},
  year={2016},
  publisher={Springer}
}

@article{pacchiardi2021generalized,
  title={{Generalized Bayesian likelihood-free inference using scoring rules estimators}},
  author={Pacchiardi, Lorenzo and Dutta, Ritabrata},
  journal={arXiv:2104.03889},
  year={2021}
}

@book{berlinet2011reproducing,
  title={Reproducing kernel Hilbert spaces in probability and statistics},
  author={Berlinet, Alain and Thomas-Agnan, Christine},
  year={2011},
  publisher={Springer Science \& Business Media}
}

@article{sriperumbudur2010hilbert,
  title={Hilbert space embeddings and metrics on probability measures},
  author={Sriperumbudur, Bharath K and Gretton, Arthur and Fukumizu, Kenji and Sch{\"o}lkopf, Bernhard and Lanckriet, Gert RG},
  journal={The Journal of Machine Learning Research},
  volume={11},
  pages={1517--1561},
  year={2010},
  publisher={JMLR. org}
}

@article{muller1997integral,
  title={Integral probability metrics and their generating classes of functions},
  author={M{\"u}ller, Alfred},
  journal={Advances in Applied Probability},
  volume={29},
  number={2},
  pages={429--443},
  year={1997},
  publisher={Cambridge University Press}
}

@inproceedings{dziugaite2015training,
  title={Training generative neural networks via maximum mean discrepancy optimization},
  author={Dziugaite, Gintare Karolina and Roy, Daniel M and Ghahramani, Zoubin},
  booktitle={Proceedings of the Thirty-First Conference on Uncertainty in Artificial Intelligence},
  pages={258--267},
  year={2015}
}

@inproceedings{li2015generative,
  title={Generative moment matching networks},
  author={Li, Yujia and Swersky, Kevin and Zemel, Rich},
  booktitle={International Conference on Machine Learning},
  pages={1718--1727},
  year={2015},
  organization={PMLR}
}

@article{sethuraman1994constructive,
  title={A constructive definition of Dirichlet priors},
  author={Sethuraman, Jayaram},
  journal={Statistica sinica},
  pages={639--650},
  year={1994},
  publisher={JSTOR}
}

@article{muliere1996bayesian,
  title={{Bayesian nonparametric predictive inference and bootstrap techniques}},
  author={Muliere, Pietro and Secchi, Piercesare},
  journal={Annals of the Institute of Statistical Mathematics},
  volume={48},
  number={4},
  pages={663--673},
  year={1996},
  publisher={Springer}
}

@article{ishwaran2002exact,
  title={Exact and approximate sum representations for the Dirichlet process},
  author={Ishwaran, Hemant and Zarepour, Mahmoud},
  journal={Canadian Journal of Statistics},
  volume={30},
  number={2},
  pages={269--283},
  year={2002},
  publisher={Wiley Online Library}
}

@article{knoblauch2020robust,
  title={{Robust Bayesian inference for discrete outcomes with the total variation distance}},
  author={Knoblauch, Jeremias and Vomfell, Lara},
  journal={arXiv:2010.13456},
  year={2020}
}

@article{bradbury2020jax,
  title={JAX: composable transformations of Python+ NumPy programs, 2018},
  author={Bradbury, James and Frostig, Roy and Hawkins, Peter and Johnson, Matthew James and Leary, Chris and Maclaurin, Dougal and Wanderman-Milne, Skye},
  journal={URL http://github. com/google/jax},
  volume={4},
  pages={16},
  year={2020}
}

@article{riesselman2018deep,
  title={Deep generative models of genetic variation capture the effects of mutations},
  author={Riesselman, Adam J and Ingraham, John B and Marks, Debora S},
  journal={Nature methods},
  volume={15},
  number={10},
  pages={816--822},
  year={2018},
  publisher={Nature Publishing Group}
}

@incollection{huber1992robust,
  title={Robust estimation of a location parameter},
  author={Huber, Peter J},
  booktitle={Breakthroughs in statistics},
  pages={492--518},
  year={1992},
  publisher={Springer}
}

@article{kingma2014adam,
  title={Adam: a method for stochastic optimization},
  author={Kingma, Diederik P and Ba, Jimmy},
  journal={arXiv preprint arXiv:1412.6980},
  year={2014}
}

@article{gretton2012kernel,
  title={A kernel two-sample test},
  author={Gretton, Arthur and Borgwardt, Karsten M and Rasch, Malte J and Sch{\"o}lkopf, Bernhard and Smola, Alexander},
  journal={The Journal of Machine Learning Research},
  volume={13},
  number={1},
  pages={723--773},
  year={2012},
  publisher={JMLR. org}
}

@article{Fournier2015,
 author = {Fournier, N. and Guillin, A.},
 year = {2015},
 title = {On the rate of convergence in Wasserstein distance of the empirical measure},
 pages = {707--738},
 volume = {162},
 number = {3-4},
 issn = {1432-2064},
 journal = {Probability Theory and Related Fields},
 doi = {10.1007/s00440-014-0583-7}
}

@article{Price2017,
author = {Price, L. F. and Drovandi, C. C. and Lee, A. and Nott, D. J.},
journal = {Journal of Computational and Graphical Statistics},
title = {{Bayesian synthetic likelihood}},
volume = {8600},
year = {2017}
}

@inproceedings{cherief2020mmd,
  title={{MMD-Bayes: robust Bayesian estimation via maximum mean discrepancy}},
  author={Ch{\'e}rief-Abdellatif, Badr-Eddine and Alquier, Pierre},
  booktitle={Symposium on Advances in Approximate Bayesian Inference},
  pages={1--21},
  year={2020},
  organization={PMLR}
}

@inproceedings{DBLP:conf/iclr/SutherlandTSDRS17,
  author    = {Danica J. Sutherland and
               Hsiao{-}Yu Tung and
               Heiko Strathmann and
               Soumyajit De and
               Aaditya Ramdas and
               Alexander J. Smola and
               Arthur Gretton},
  title     = {Generative models and model criticism via optimized maximum mean discrepancy},
  booktitle = {International Conference on Learning Representations},
  year      = {2017},
}

@article{prangle2020gk,
  title={gk: an R Package for the g-and-k and Generalised g-and-h distributions},
  author={Prangle, Dennis},
  journal={The R Journal},
  volume={12},
  number={1},
  pages={7--20},
  year={2020}
}

@article{bonassi2015sequential,
  title={{Sequential Monte Carlo with adaptive weights for approximate Bayesian computation}},
  author={Bonassi, Fernando V and West, Mike},
  journal={Bayesian Analysis},
  volume={10},
  number={1},
  pages={171--187},
  year={2015},
  publisher={International Society for Bayesian Analysis}
}

@article{bonassi2011bayesian,
  title={{Bayesian learning from marginal data in bionetwork models}},
  author={Bonassi, Fernando V and You, Lingchong and West, Mike},
  journal={Statistical applications in genetics and molecular biology},
  volume={10},
  number={1},
  year={2011},
  publisher={De Gruyter}
}

@article{Parr1980,
  title={Minimum distance and robust estimation},
  author={Parr, W C and Schucany, W R},
  journal={Journal of the American Statistical Association},
  volume={75},
  number={371},
  pages={616--624},
  year={1980},
  publisher={Taylor \& Francis}
}

@book{bernardo2009bayesian,
  title={{Bayesian theory}},
  author={Bernardo, Jos{\'e} M and Smith, Adrian FM},
  volume={405},
  year={2009},
  publisher={John Wiley \& Sons}
}

@article{Andrieu2009,
author = {Andrieu, Christophe and Roberts, Gareth O.},
journal = {Annals of Statistics},
number = {2},
pages = {697--725},
title = {{The pseudo-marginal approach for efficient Monte Carlo computations}},
volume = {37},
year = {2009}
}

@inproceedings{papamakarios2016fast,
  title={{Fast $\varepsilon$-free inference of simulation models with Bayesian conditional density estimation}},
  author={Papamakarios, George and Murray, Iain},
  booktitle={Advances in Neural Information Processing Systems},
  pages={1028--1036},
  year={2016}
}

@InProceedings{papamakarios19a, 
title = {{Sequential neural likelihood: fast likelihood-free inference with autoregressive flows}},
author = {Papamakarios, George and Sterratt, David and Murray, Iain}, 
pages = {837--848}, 
year = {2019},
volume = {89}, 
series = { International Conference on Artificial Intelligence and Statistics}
}

@inproceedings{safeBayes,
  title={{The safe Bayesian}},
  author={Gr{\"u}nwald, Peter},
  booktitle={International Conference on Algorithmic Learning Theory},
  pages={169--183},
  year={2012},
  organization={Springer}
}

@article{flamary2021pot,
  title={Pot: Python optimal transport},
  author={Flamary, R{\'e}mi and Courty, Nicolas and Gramfort, Alexandre and Alaya, Mokhtar Zahdi and Boisbunon, Aur{\'e}lie and Chambon, Stanislas and Chapel, Laetitia and Corenflos, Adrien and Fatras, Kilian and Fournier, Nemo and others},
  journal={Journal of Machine Learning Research},
  volume={22},
  number={78},
  pages={1--8},
  year={2021}
}

@article{lyddon2019general,
  title={{General Bayesian updating and the loss-likelihood bootstrap}},
  author={Lyddon, SP and Holmes, CC and Walker, SG},
  journal={Biometrika},
  volume={106},
  number={2},
  pages={465--478},
  year={2019},
  publisher={Oxford University Press}
}

@article{catoniPACBound,
  title={{Pac-Bayesian supervised classification: the thermodynamics of statistical learning}},
  author={Catoni, O},
  journal={Institute of Mathematical Statistics Lecture Notes—Monograph Series 56},
  year={2007}
}

@inproceedings{ShaweTaylor,
  title={A {PAC} analysis of a {B}ayesian estimator},
  author={Shawe-Taylor, John and Williamson, Robert C},
  booktitle={Proceedings of the tenth annual conference on Computational learning theory},
  volume={6},
  number={09},
  pages={2--9},
  year={1997}
}

@article{virtanen2020scipy,
  title={SciPy 1.0: fundamental algorithms for scientific computing in Python},
  author={Virtanen, Pauli and Gommers, Ralf and Oliphant, Travis E and Haberland, Matt and Reddy, Tyler and Cournapeau, David and Burovski, Evgeni and Peterson, Pearu and Weckesser, Warren and Bright, Jonathan and others},
  journal={Nature methods},
  volume={17},
  number={3},
  pages={261--272},
  year={2020},
  publisher={Nature Publishing Group}
}

@article{galvani2021bayesian,
  title={{A Bayesian nonparametric learning approach to ensemble models using the proper Bayesian bootstrap}},
  author={Galvani, Marta and Bardelli, Chiara and Figini, Silvia and Muliere, Pietro},
  journal={Algorithms},
  volume={14},
  number={1},
  pages={11},
  year={2021},
  publisher={Multidisciplinary Digital Publishing Institute}
}

@article{thomas2020split,
  title={Split-BOLFI for for misspecification-robust likelihood free inference in high dimensions},
  author={Thomas, Owen and Pesonen, Henri and S{\'a}-Le{\~a}o, Raquel and de Lencastre, Herm{\'\i}nia and Kaski, Samuel and Corander, Jukka},
  journal={arXiv preprint arXiv:2002.09377},
  year={2020}
}

@article{thomas2020generalised,
  title={Generalised Bayes Updates with $ f $-divergences through Probabilistic Classifiers},
  author={Thomas, Owen and Pesonen, Henri and Corander, Jukka},
  journal={arXiv preprint arXiv:2007.04358},
  year={2020}
}

@article{pompe2021introducing,
  title={Introducing prior information in Weighted Likelihood Bootstrap with applications to model misspecification},
  author={Pompe, Emilia},
  journal={arXiv preprint arXiv:2103.14445},
  year={2021}
}

@article{huggins2019robust,
  title={Robust inference and model criticism using bagged posteriors},
  author={Huggins, Jonathan H and Miller, Jeffrey W},
  journal={arXiv preprint arXiv:1912.07104},
  year={2019}
}

@article{loaiza2021focused,
  title={{Focused Bayesian prediction}},
  author={Loaiza-Maya, Ruben and Martin, Gael M and Frazier, David T and others},
  journal={Journal of Applied Econometrics},
  volume={36},
  number={5},
  pages={517--543},
  year={2021},
  publisher={John Wiley \& Sons, Ltd.}
}


\clearpage
\appendix

\thispagestyle{empty}

\onecolumn \makesupplementtitle

In Appendix \ref{app:notation} we summarise the notation used throughout the paper. In Appendix \ref{app:proofs} we prove all theoretical results. In Appendix \ref{app:experiments} we provide details on the experiments introduced in the main text and in Appendix \ref{app:additional-exps} we provide some additional experiments.

\section{NOTATION} \label{app:notation}

In the following section, we recall the notation used in the paper: 

\vspace{-10mm}
\printunsrtglossary[type=symbols, title={}]

\section{PROOFS} \label{app:proofs}

We now provide proofs for all theoretical results presented in the main text as well as some additional theory. 
We start with the necessary background material on the MMD that we use in our proofs. 
We continue by proving our generalisation error bound both in the case where the DP posterior sample is represented exactly by an infinite sum of the stick-breaking process (Theorem \ref{gen-error}) and using the approximation of algorithm \ref{alg:MMD_posterior_bootstrap} (Theorem \ref{gen_error_bound}). 
We also provide the corresponding corollaries for the well-specified case (Corollaries \ref{cor-gen-error-well} and \ref{cor-gen-error-well-stick}) for reference. 
We then proceed by proving the general posterior consistency result of Theorem \ref{post_consistency_miss}; along with the special case of a well-specified model  (Corollary \ref{post_consistency_well}). 
Finally, we prove the results in the case of a contaminated data generating process in both the misspecified and well-specified case (Corollaries \ref{cor-contam-mis} and \ref{cor-contam-well}). 
Recall that throughout all our theoretical results we impose Standing Assumptions \ref{as1} and \ref{as2}.

\subsection{MMD Through Kernel Mean Embeddings}
In the following proofs we will use an equivalent definition of the MMD through  \textit{kernel mean embeddings}. Consider an RKHS $\Hk$ indexed by a reproducing kernel $k$. We say $k$ is a reproducing kernel if (i) $k(\cdot, x) \in \Hk$, $\forall x \in \X$, (ii) $\left<f, k(\cdot, x) \right>_{\Hk} = f(x)$ $\forall x \in \X$ and $\forall f \in \Hk$ (reproducing property); and inner product $\left< \cdot, \cdot \right>_{\Hk}$; see \textcite{berlinet2011reproducing}. For a function $f \in \Hk$ the mean embedding $\mu_{\P} \in \Hk$ with respect to probability measure $\P \in \mathcal{P}$ is defined as  
\begin{talign*}
    \mu_{\P}(\cdot) := \mathbb{E}_{X \sim \P}[k(X, \cdot)] \in \Hk 
\end{talign*}
and satisfies 
\begin{talign*}
 \mathbb{E}_{X \sim \P}[f(X)] = \left< f, \mu_{\P} \right>_{\Hk}.
\end{talign*}
The MMD between probability measures $\P$ and $\Q$ takes the form \parencite[see][Lemma 4]{gretton2012kernel}
\begin{talign*}
\MMD(\P,\Q) = \| \mu_{\P} - \mu_{\Q} \|_{\Hk}.
\end{talign*}

\subsection{Generalisation Error for the Stick-Breaking Process} \label{app:gen-error-exact}

We start with a bound of the generalisation error for the case where the DP sample is represented \textit{exactly}---i.e., by the infinite sum obtained from the stick-breaking process. 
We denote by $\nu$ the  probability measure on $\PX$ corresponding to the stick-breaking process given as
\begin{talign} 
   w_{1:\infty} &\sim \GEM(\alpha'), \quad \alpha' := \alpha + n, \label{eq:gem-weights} \\
   z_{1: \infty} &\simiid \F' :=  \frac{\alpha}{\alpha+n} \F + \frac{n}{n+\alpha} \Pn, \label{eq:post-cent-measure} \\
    \P & = \sum_{i=1}^{\infty} w_i \delta_{z_i} \sim \nplexactposterior. \nonumber
\end{talign}
Note that instead of writing the expectation directly over  $\nplexactposterior$, we instead often use the separate expectations $\E_{w_{1:\infty} \sim \GEM(\alpha')}$ and $\E_{z_{1:\infty} \sim \F'}$ induced by $\nplexactposterior$ in the proofs below.

Before stating and proving the main result of this section, we provide three lemmas that bound the expected MMD between the true data generating mechanism $\truedist$ and $\Pn$, the centering measure of the DP posterior $\F'$ and $\P$ and lastly $\F'$ and $\Pn$.
The reason for this is that the main proof will use a decomposition of the MMD using the triangle inequality; and
the following technical Lemmas bound three of the trickier terms arising from said triangle inequality.

\begin{lemma} \label{lemma7.1}
For $x_{1:n} \simiid \truedist$ we have 
\begin{talign*}
\E_{x_{1:n} \simiid \truedist} \left[\MMD(\Pn, \truedist) \right] \leq \frac{1}{\sqrt{n}}.
\end{talign*}
where $\Pn$ denotes the empirical measure of the sample data, i.e. $\Pn = \frac{1}{n} \sum_{i=1}^n \delta_{x_i}$.
\end{lemma}
\begin{proof}
Proof follows directly from the proof of \textcite{cherief2019finite}, Lemma 7.1 and using the Jensen's inequality to obtain 
\begin{talign*}
    \E_{x_{1:n} \simiid \truedist}\left[ \MMD(\Pn, \truedist)\right] \leq \sqrt{\E_{x_{1:n} \simiid \truedist}\left[ \MMD^2(\Pn, \truedist)\right]} \leq \frac{1}{\sqrt{n}}.
\end{talign*}
\end{proof}
\begin{lemma}\label{mmd_f_gn} 
Let $\P = \sum_{i=1}^{\infty} w_i \delta_{z_i} \sim \nplexactposterior$ denote a sample from the DP posterior, where $z_{1:\infty} \simiid \F'$ and $w_{1:\infty} \sim \GEM(\alpha') $, then
\begin{talign*}
    \E_{w_{1:\infty} \sim \GEM(\alpha')} \left[ \E_{z_{1:\infty} \sim \F'} \left[ \MMD(\P, \F')  \right] \right] \leq \frac{1}{\sqrt{\alpha + n + 2}}.
\end{talign*}
\end{lemma}
\begin{proof}
First note that 
\begin{talign*}
    \MMD^2( \P, \F') &= \| \mu_{\P} - \mu_{\F'} \|_{\Hk}^2 \\
    &= \| \sum_{i=1}^{\infty} w_i k(z_i, \cdot) - \mu_{\F'} \|_{\Hk}^2 \\
    &= \| \sum_{i=1}^{\infty} w_i \left[ k(z_i, \cdot) - \mu_{\F'}  \right] \|_{\Hk}^2 \\
    &= \sum_{i=1}^{\infty} w_i^2 \| k(z_i, \cdot) - \mu_{\F'} \|_{\Hk}^2 + 2 \sum_{i \neq j} w_i w_j \left< k(z_i, \cdot) - \mu_{\F'}, k(z_j, \cdot) - \mu_{\F'} \right>
\end{talign*}
Note that since $z_{1:\infty} \simiid \F'$ we have that for any $i \neq j$:
\begin{talign}
    & \E_{z_i, z_j \simiid \F'} \left[ \left< k(z_i, \cdot) - \mu_{\F'}, k(z_j, \cdot) - \mu_{\F'} \right> \right] \nonumber \\
    &= \E_{z_i, z_j \simiid \F'} \left[\left< k(z_i, \cdot) - \E_{z_i \sim \F'}[k(z_i, \cdot)], k(z_j, \cdot) - \E_{z_j \sim \F'}[(k(z_j, \cdot)] \right> \right] \nonumber \\
    &= \E_{z_i, z_j \simiid \F'}[k(z_i, z_j)] - \E_{z_i \sim \F'}[k(z_i, \cdot)] \E_{z_j \sim \F'}[k(z_j, \cdot)] \notag \nonumber\\
     &\qquad - \E_{z_j \sim \F'}[k(z_j, \cdot)] \E_{z_i \sim \F'}[k(z_i, \cdot)] + \E_{z_j \sim \F'}[k(z_j, \cdot)] \E_{z_i \sim \F'}[k(z_i, \cdot)] \nonumber\\
    &= \E_{z_i, z_j \simiid \F'}[k(z_i, z_j)] - \E_{z_i, z_j \simiid \F'}[k(z_i, z_j)] - \E_{z_i, z_j \simiid \F'}[k(z_i, z_j)] + \E_{z_i, z_j \simiid \F'}[k(z_i, z_j)] \nonumber \\
    &= 0. \label{eq:zeroeq}
\end{talign}
Moreover, for any $i = 1, 2, \dots$
\begin{talign}
    \E_{z_i \sim \F'} \left[ \| k(z_i, \cdot) - \mu_{\F'} \|_{\Hk}^2    \right] &= \E_{z_i \sim \F'} \left[ \| k(z_i, \cdot) - \E_{z_i \sim \F'}[ k(z_i, \cdot)] \|_{\Hk}^2 \right] \nonumber \\
    &= \E_{z_i \sim \F'} \left[ \| k(z_i, \cdot) \|_{\Hk}^2 - 2 \left< k(z_i, \cdot), \E_{z_i \sim \F'}[k(z_i, \cdot)] \right> + \|\E_{z_i \sim \F'} [k(z_i, \cdot)] \|_{\Hk}^2 \right] \nonumber\\
    &= \E_{z_i \sim \F'} \left[ \| k(z_i, \cdot) \|_{\Hk}^2 \right] - 2 \left< \E_{z_i \sim \F'} [k(z_i, \cdot)], \E_{z_i \sim \F'}[k(z_i, \cdot)] \right> + \|\E_{z_i \sim \F'} [k(z_i, \cdot)] \|_{\Hk}^2 \nonumber\\
    &= \E_{z_i \sim \F'} \left[ \| k(z_i, \cdot) \|_{\Hk}^2 \right] - 2   \|\E_{z_i \sim \F'} [k(z_i, \cdot)] \|_{\Hk}^2 + \|\E_{z_i \sim \F'} [k(z_i, \cdot)] \|_{\Hk}^2 \nonumber\\
    &= \E_{z_i \sim \F'} [\| k(z_i, \cdot) \|_{\Hk}^2] - \|\E_{z_i \sim \F'} [k(z_i, \cdot)] \|_{\Hk}^2 \nonumber\\
    &\leq \E_{z_i \sim \F'} [\| k(z_i, \cdot) \|_{\Hk}^2] \nonumber\\
    &= \E_{z_i \sim \F'} [ | k(z_i, z_i) |] \label{eq:ineq1}
\end{talign}
where the last equality follows from the fact that 
$$
\| k(z_i, \cdot) \|_{\Hk}^2 = |\left< k(z_i, \cdot), k(\cdot, z_i) \right>| = |k(z_i, z_i)|  
$$
using the reproducing property of the RKHS which says that $\forall f \in \Hk \left<f, k(x,\cdot) \right>_{\Hk} = f(x)$. We then obtain:
\begin{talign*}
    \E_{w_{1:\infty} \sim \GEM(\alpha')} \left[ \E_{z_{1:\infty} \sim \F'} \left[ \MMD^2 (\P, \F')  \right] \right] &\leq   \sum_{i=1}^{\infty}  \E_{w_{1:\infty} \sim \GEM(\alpha')} [w_i^2] \E_{z_i \sim \F'} [\| k(z_i, \cdot) - \mu_{\F'} \|_{\Hk}^2 ] + 2 \sum_{i \neq j} 0 \\
    &\leq  \sum_{i=1}^{\infty}  \E_{w_{1:\infty} \sim \GEM(\alpha')} [w_i^2] \E_{z_i \sim \F'}  [| k(z_i, z_i) | ]  \\
    &\leq \sum_{i=1}^{\infty} \E_{w_{1:\infty} \sim \GEM(\alpha')} [ w_i^2].
\end{talign*}
The first inequality above follows from equation \eqref{eq:zeroeq}, the second inequality follows from equation \eqref{eq:ineq1}, and the last inequality follows from the boundedness of the kernel in standing assumption \ref{as1}.

From the properties of the GEM and Beta distributions we have that since $w_k = \beta_k \prod_{i=1}^{k-1} (1 - \beta_i)$ where $\beta_{1:\infty} \simiid \text{Beta}(1, \alpha + n)$ then 
\begin{talign*}
    \E_{w_{1:\infty} \sim \GEM(\alpha')}[w_i] = \frac{(\alpha+n)^{i-1}}{(1 + \alpha + n)^i} \qquad \text{ and }\qquad 
     \E_{w_{1:\infty} \sim \GEM(\alpha')}[w_i^2] = \frac{2 (\alpha +n)^{i-1}}{(\alpha + n + 2)^i (\alpha + n + 1)}.
\end{talign*}
Therefore 
\begin{talign*}
      \E_{w_{1:\infty} \sim \GEM(\alpha')} \left[ \E_{z_{1:\infty} \simiid \F'} \left[ \MMD^2 (\P, \F')  \right] \right] &\leq \sum_{i=1}^{\infty} \E_{w_{1:\infty} \sim \GEM(\alpha')} [ w_i^2] \\
     &=  \frac{2}{(\alpha + n + 2) (\alpha + n + 1)}\sum_{i=1}^{\infty} \left(\frac{\alpha + n}{\alpha + n + 2}\right)^{i-1} \\
     &= \frac{2}{(\alpha + n + 2) (\alpha + n + 1)} \frac{1}{1 - \frac{\alpha + n}{\alpha + n + 2}} \\
     &= \frac{1}{\alpha + n + 2}
\end{talign*}
Finally, by Jensen's inequality, 
\begin{talign*}
     \E_{w_{1:\infty} \sim \GEM(\alpha')} \left[ \E_{z_{1:\infty} \simiid \F'} \left[ \MMD(\P, \F')  \right] \right] \leq \sqrt{ \E_{w_{1:\infty} \sim \GEM(\alpha')} \left[ \E_{z_{1:\infty} \simiid \F'} \left[ \MMD^2 (\P, \F')  \right] \right]} \leq \frac{1}{\sqrt{\alpha + n + 2}}.
\end{talign*}
\end{proof}

\begin{lemma}\label{mmd_pn_gn}
Let $\Pn$ denote the empirical measure $\Pn = \frac{1}{n} \sum_{i=1}^n \delta_{x_i}$ then 
\begin{talign*}
    \E_{x_{1:n} \simiid \truedist} \left[ \E_{z_{1:\infty} \simiid \F'} \left[ \MMD ( \Pn, \F') \right] \right] \leq \frac{2 \alpha}{\alpha + n}.
\end{talign*}
\end{lemma}
\begin{proof}
By definition of $\F'$ and linearity of expectations, it follows that $\E_{\F'}[\cdot] = \frac{\alpha}{\alpha+n} \E_{\F}[\cdot] + \frac{n}{n+\alpha} \E_{\Pn}[\cdot]$ and hence
\begin{talign} \label{eq:muFprime}
\mu_{\F'} = \E_{\F'} [ k(z, \cdot)] = \frac{\alpha}{\alpha+n} \E_{\F}[k(z,\cdot)] + \frac{n}{n+\alpha} \E_{\Pn}[k(z, \cdot)] = \frac{\alpha}{\alpha+n} \mu_{\F} + \frac{n}{n+\alpha} \mu_{\Pn}.
\end{talign}
Therefore, 
\begin{talign*}
    & \E_{x_{1:n} \simiid \truedist} \left[ \E_{z_{1:\infty} \simiid \F'} \left[ \MMD^2(\Pn, \F') \right] \right] \\
    &= \E_{x_{1:n} \simiid \truedist} \left[ \E_{z_{1:\infty} \simiid \F'} \left[ \| \mu_{\Pn} -  \mu_{\F'}  \|^2_{\Hk} \right] \right] \\
    &= \E_{x_{1:n} \simiid \truedist} \left[ \E_{z_{1:\infty} \simiid \F'} \left[ \left\| \mu_{\Pn} -  \frac{\alpha}{\alpha+n} \mu_{\F} - \frac{n}{n + \alpha} \mu_{\Pn} \right\|^2_{\Hk}  \right] \right]\\
    &= \frac{\alpha^2}{(\alpha + n)^2 } \E_{x_{1:n} \simiid \truedist} \left[ \E_{z_{1:\infty} \simiid \F'} \left[ \| \mu_{\Pn} - \mu_{\F} \|^2_{\Hk} \right] \right] \\
    &= \frac{\alpha^2}{(\alpha + n)^2 } \E_{x_{1:n} \simiid \truedist} \left[ \E_{z_{1:\infty} \simiid \F'} \left[ \left\| \frac{1}{n} \sum_{i=1}^n k(x_i, \cdot) \right\|^2_{\Hk} - \frac{2}{n} \sum_{i=1}^n \left< k(x_i, \cdot), \E_{z \sim \F} [ k(z, \cdot) ] \right>+ \left\| \mu_{\F} \right\|^2_{\Hk} \right] \right]\\
    &= \frac{\alpha^2}{(\alpha + n)^2 } \E_{x_{1:n} \simiid \truedist} \left[ \E_{z_{1:\infty} \simiid \F'} \left[\frac{1}{n^2} \sum_{i=1}^n k(x_i, x_i) + \frac{2}{n^2} \sum_{i\neq j} k(x_i, x_j)  - \frac{2}{n} \sum_{i=1}^n  \E_{z \sim \F} [k(x_i, z)] \right. \right. \\  &\quad \quad \left. \left. + \| \E_{z \sim \F} [k(z,\cdot) ]\|^2_{\Hk} \right] \right]\\
    &= \frac{\alpha^2}{(\alpha + n)^2 } \E_{x_{1:n} \simiid \truedist} \left[ \E_{z_{1:\infty} \simiid \F'} \left[ \frac{1}{n^2} \sum_{i=1}^n k(x_i, x_i) + \frac{2}{n^2} \sum_{i\neq j} k(x_i, x_j) - \frac{2}{n} \E_{z \sim \F} \left[\sum_{i=1}^n   k(x_i, z) \right] \right. \right. \\   &\quad \quad \left. \left. + \E_{z,z' \simiid \F} [k(z,z')]\right] \right]\\
    &\leq \frac{\alpha^2}{(\alpha + n)^2 } \E_{x_{1:n} \simiid \truedist} \left[ \E_{z_{1:\infty} \simiid \F'} \left[ \frac{1}{n} + \frac{n(n-1)}{n^2} + 2 + 1 \right] \right]\\
    &= \frac{4 \alpha^2}{(\alpha + n)^2}.
\end{talign*}
The first three equalities above follow from the definition of the MMD in terms of the kernel mean embeddings and equation \eqref{eq:muFprime}. The following three equalities use the definition of the kernel mean embedding as well as the characteristic property which ensures that 
\begin{talign*}
\| \E_{z \sim \F} [k(z,\cdot) ]\|^2_{\Hk} &= \left< \E_{z \sim \F} k(z, \cdot),  \E_{z \sim \F} k(z, \cdot) \right>_{\Hk} \\
&= \left< \E_{z \sim \F} k(z, \cdot),  \E_{z' \sim \F} k(\cdot, z') \right>_{\Hk} \\
&= \E_{z,z' \sim \F} \left<  k(z, \cdot), k(\cdot, z') \right>_{\Hk} \\
&= \E_{z,z' \sim \F} [k(z,z')].
\end{talign*} 
For the inequality above we used again the standing assumption \ref{as1} about the boundedness of the kernel. To conclude, by Jensen's inequality we have
\begin{talign*}
    \E_{x_{1:n} \simiid \truedist} \left[ \E_{z_{1:\infty} \sim \F'} \left[ \MMD ( \Pn, \F') \right] \right] \leq \sqrt{\E_{x_{1:n} \simiid \truedist} \left[ \E_{z_{1:\infty} \sim \F'} \left[ \MMD^2(\Pn, \F') \right] \right]} \leq \frac{2 \alpha}{\alpha + n}.
\end{talign*}
\end{proof}
We can now formulate and prove the generalisation error bound below.
\begin{theorem} \label{gen-error}
Assume $x_{1:n} \simiid \truedist$ and let $\P$ be a sample from the exact DP posterior with law $\nplexactposterior$. Then
\begin{talign*}
    &\E_{x_{1:n} \simiid \truedist} \left[ \E_{\P \sim \nplexactposterior} \left[ \MMD( \truedist, \P_{\theta^\ast(\P)}) \right] \right] 
    \leq  \inf_{\theta \in \Theta}  \MMD( \P_{\theta}, \truedist) 
    + \frac{2}{\sqrt{n}} 
    + \frac{2}{\sqrt{\alpha+n+2}} 
    + \frac{4\alpha}{\alpha + n}. 
\end{talign*}
\end{theorem}
\begin{proof}
 Using the triangle inequality we have that for any $\theta \in \Theta$:

\begin{talign*}
    \MMD( \truedist, \P_{\theta^\ast(\P)}) &\leq \MMD(\P_{\theta^\ast(\P)}, \P) + \MMD(\truedist, \P) \\
    &\leq \MMD( \P_{\theta}, \P) +  \MMD(\truedist, \P)  \\
    &\leq \MMD( \P_{\theta}, \truedist) + 2 \MMD(\truedist, \P) \\
    &\leq \MMD( \P_{\theta}, \truedist) + 2 \MMD(\Pn, \truedist) + 2 \MMD(\P, \Pn) \\
    &\leq \MMD( \P_{\theta}, \truedist) + 2 \MMD(\Pn, \truedist) + 2 \MMD(\P, \F') + 2 \MMD(\Pn, \F').
\end{talign*}
For the first step above we used the triangle inequality and for the second we used the fact that since $\theta^\ast = \arg\inf_{\theta \in \Theta} \MMD(\truedist, \P_{\theta})$ it follows that $\MMD(\P_{\theta^\ast(\P)}, \P) \leq \MMD(\Ptheta, \P)$ for any $\theta \in \Theta$. The remaining inequalities are obtained by reapplying the triangle inequality on $\MMD(\Ptheta, \P)$, $\MMD(\truedist, \P)$ and $\MMD(\P, \Pn)$ respectively. Since the above is true for any $\theta \in \Theta$ it follows that
\begin{talign*}
\MMD( \truedist, \P_{\theta^\ast(\P)}) \leq \inf_{\theta \in \Theta}  \MMD( \P_{\theta}, \truedist) + 2 \MMD(\Pn, \truedist) + 2 \MMD(\P, \F') + 2 \MMD(\Pn, \F').
\end{talign*}
Taking expectations on both sides we obtain 
\begin{talign*}
    \E_{x_{1:n} \simiid \truedist} \left[ \E_{w_{1:\infty} \sim \GEM(\alpha')} \left[ \E_{z_{1:\infty} \simiid \F'} \left[ \MMD( \P, \P_{\theta^\ast(\P)}) \right] \right] \right] \leq  &\inf_{\theta \in \Theta}  \MMD( \P_{\theta}, \truedist) + 2 \E_{x_{1:n} \simiid \truedist} \left[ \MMD(\Pn, \truedist)  \right] \\
    &+ 2 \E_{w_{1:\infty} \sim \GEM(\alpha')} \left[ \E_{z_{1:\infty} \simiid \F'} \left[ \MMD(\P, \F')  \right] \right] \\
    &+ 2 \E_{x_{1:n} \simiid \truedist} \left[ \E_{z_{1:\infty} \simiid \F'} \left[ \MMD(\Pn, \F') \right] \right]. 
\end{talign*}
The result now follows by Lemmas \ref{lemma7.1}, \ref{mmd_f_gn} and \ref{mmd_pn_gn}.

\end{proof}

The well-specified case is an immediate consequence of the Theorem:

\begin{corollary} \label{cor-gen-error-well-stick}
Suppose that the model is well-specified, i.e. $\inf_{\theta \in \Theta}  \MMD( \P_{\theta}, \truedist) = 0$ then
\begin{talign*}
    \E_{x_{1:n} \simiid \truedist} \left[ \E_{\P \sim \nplexactposterior} \left[ \MMD( \truedist, \P_{\theta^\ast(\P)}) \right] \right] 
    \leq \frac{2}{\sqrt{n}} 
    + \frac{2}{\sqrt{\alpha+n+2}} 
    + \frac{4\alpha}{\alpha + n}. 
\end{talign*}
\end{corollary}

\begin{proof}
The corollary follows directly from Theorem \ref{gen-error} by setting $\inf_{\theta \in \Theta}  \MMD( \P_{\theta}, \truedist) = 0$.
\end{proof}

\subsection{Theorem \ref{gen_error_bound}}
We start with a lemma bounding the expected MMD between the empirical measure and the approximated DP posterior sample. Recall that the probability measure $\nplposterior$  refers to the law on $\PX$ induced by the sampling process:
\begin{talign}
    (w_{1:n}, \tilde{w}_{n+T}) &\sim \Dir\left(1,\dots,1,\frac{\alpha}{T}, \dots, \frac{\alpha}{T}\right) \label{eq:dir-weights} \\ 
    \tilde{x}_{1:T} &\simiid \F \label{eq:fake-data}\\ 
    \P := \sum_{i=1}^n w_i \delta_{x_i} + \sum_{k=1}^T \tilde{w}_k \delta_{\tilde{x}_k} &\sim \nplposterior. \nonumber
\end{talign}
For clarity we use the individual expectations $\E_{(w_{1:n}, \tilde{w}_{n+T}) \sim \Dir\left(1,\dots,1,\frac{\alpha}{T}, \dots, \frac{\alpha}{T}\right)}$, which we denote by $\E_{w \sim \Dir}$, and $\E_{\tilde{x}_{1:T} \simiid \F}$ arising from Equations \eqref{eq:dir-weights}-\eqref{eq:fake-data} in the proofs below---rather than a single expectation over $\nplposterior$.
\begin{lemma} \label{ptildephat}
Let $\P \sim \nplposterior$ denote a single draw from the approximated DP posterior, i.e.
\begin{talign*} 
\P = \sum_{i=1}^n w_i \delta_{x_i} + \sum_{k=1}^T \tilde{w}_k \delta_{\tilde{x}_k}.
\end{talign*}
Then
\begin{talign*}
    \E_{x_{1:n} \simiid \truedist} \left[\E_{\P \sim \nplposterior} \left[\MMD( \P, \Pn) \right] \right] \leq  2\sqrt{\frac{\alpha(1+\alpha)}{(\alpha+n)(\alpha+n+1)}}
\end{talign*}
\end{lemma}

\begin{proof}

The mean embedding for $\P$ is $\mu_{\P} = \sum_{i=1}^n w_i k(x_i, \cdot) + \sum_{j=1}^T \tilde{w}_j k(\tilde{x}_j, \cdot)$. Hence using the triangle inequality:
\begin{talign}
    \MMD \left( \P, \Pn \right) &= \| \mu_{\P} - \mu_{\Pn} \|_{\Hk} \nonumber \\
    &= \left\|  \sum_{i=1}^n w_i k(x_i, \cdot) + \sum_{j=1}^T \tilde{w}_j k(\tilde{x}_j, \cdot) - \mu_{\Pn} \right\|_{\Hk} \nonumber \\
    &\leq \left\|  \sum_{i=1}^n w_i k(x_i, \cdot) - \mu_{\Pn} \right\|_{\Hk} + \left\| \sum_{j=1}^T \tilde{w}_j k(\tilde{x}_j, \cdot) \right\|_{\Hk}. \label{eq:ineq5}
\end{talign}
To bound $\E_{x_{1:n} \simiid \truedist} \left[\E_{\P \sim \nplposterior} \left[\| \sum_{i=1}^n w_i k(x_i, \cdot) - \mu_{\Pn} \|_{\Hk} \right]\right]$, first we note that from the moments of a Dirichlet distribution we have that for any $i,j \in \{1, \dots, n\}$:
\begin{talign}
    \E_{w \sim \Dir}[w_i] = \frac{1}{n+\alpha}, \qquad \E_{w \sim \Dir}[w_i^2] = \frac{2}{(\alpha+n+1)(\alpha+n)} \qquad \text{ and } \qquad
    \E_{w \sim \Dir}[w_i w_j] &= \frac{1}{(\alpha+n) (\alpha + n + 1)}. \label{eq:exp-dir}
\end{talign}
Hence 
\begin{talign*}
    &\E_{x_{1:n} \simiid \truedist} \left[ \E_{\P \sim \nplposterior} \left[ \left\| \frac{1}{n} \sum_{i=1}^n k(x_i, \cdot) - \sum_{i=1} w_i k(x_i, \cdot) \right\|_{\Hk}^2 \right] \right]  \\
    &= \E_{x_{1:n} \simiid \truedist} \left[ \E_{\P \sim \nplposterior} \left[ \left\| \sum_{i=1}^n \left(\frac{1}{n} - w_i \right) k(x_i, \cdot) \right\|_{\Hk}^2 \right] \right]\\
    &= \E_{x_{1:n} \simiid \truedist} \left[ \E_{\P \sim \nplposterior} \left[ \sum_{i=1}^n \left(\frac{1}{n} - w_i \right)^2 | k(x_i, x_i) |   + 2 \sum_{1 \leq i < j \leq n}  \left(\frac{1}{n} - w_i \right) \left(\frac{1}{n} - w_j \right)   k(x_i, x_j)  \right] \right]\\
    &\leq  \E_{\P \sim \nplposterior} \left[\sum_{i=1}^n \left(\frac{1}{n} - w_i \right)^2 \right] + 2 \E_{x_{1:n} \simiid \truedist} \left[ \E_{\P \sim \nplposterior} \left[ \sum_{1 \leq i < j \leq n}  \left(\frac{1}{n} - w_i \right) \left(\frac{1}{n} - w_j \right) k(x_i,x_j) \right] \right]  \\
    &= \sum_{i=1}^n \left[ \frac{1}{n^2} - \frac{2}{n} \E_{w \sim \Dir}(w_i) + \E_{w \sim \Dir}(w_i^2) \right] \\ 
    &\quad \quad + 2 \sum_{1 \leq i < j \leq n} \left[\left(\frac{1}{n^2} - \frac{1}{n} \E_{w \sim \Dir}(w_i) - \frac{1}{n} \E_{w \sim \Dir}(w_j) + \E_{w \sim \Dir} (w_i w_j) \right) \mathbb{E}_{x_{1:n} \simiid \truedist} (k(x_i,x_j)) \right] \\
    &= \frac{1}{n} - \frac{2n}{n(n+\alpha)} + \frac{2n}{(n+\alpha)(n + \alpha +1)} + 2\sum_{1 \leq i < j \leq n} \left[\left(\frac{1}{n^2} - \frac{2}{n(n+\alpha)} + \frac{1}{(n+\alpha)(n+\alpha+1)} \right) \mathbb{E}_{x_{1:n} \simiid \truedist} (k(x_i,x_j))\right] \\
    &= \frac{n^2 - n + \alpha^2 + \alpha}{n(n + \alpha)(n + \alpha + 1)} + 2 \sum_{1 \leq i < j \leq n} \left[\frac{\alpha^2 + \alpha - n}{n^2(n+\alpha)(n+ \alpha + 1)}\mathbb{E}_{x_{1:n} \simiid \truedist} (k(x_i,x_j)) \right] \\
    &\leq \frac{n^2 - n + \alpha^2 + \alpha}{n(n + \alpha)(n + \alpha + 1)} + 2 \sum_{1 \leq i < j \leq n} \left[\frac{|\alpha^2 + \alpha - n|}{n^2(n+\alpha)(n+ \alpha + 1)}\right] \\
    &\leq \frac{n^2 - n + \alpha^2 + \alpha}{n(n + \alpha)(n + \alpha + 1)} + 2 \sum_{1 \leq i < j \leq n} \left[\frac{\alpha^2 + \alpha + n}{n^2(n+\alpha)(n+ \alpha + 1)}\right] \\
    &= \frac{2(n-1) + \alpha(\alpha+1)}{(n+\alpha)(n+\alpha+1)}
\end{talign*}
Here the inequalities follow again from standing assumption \ref{as1} and the result is obtained by using the expectations in \eqref{eq:exp-dir}. Hence, by Jensen's inequality, 
\begin{talign} \label{eq:ineq3}
\E_{x_{1:n} \simiid \truedist} \left[\E_{\P \sim \nplposterior} \left[\left\| \sum_{i=1}^n w_i k(x_i, \cdot) - \mu_{\Pn} \right\|_{\Hk} \right]\right]  \leq \sqrt{\frac{2(n - 1) + \alpha(\alpha+1)}{(\alpha+n) (n + \alpha + 1)}}.
\end{talign}
Similarly due to standing assumption \ref{as1}, to bound $ \E_{\P \sim \nplposterior} [\| \sum_{j=1}^T \tilde{w}_j k(\tilde{x}_j, \cdot) \|_{\Hk} ]$, we note that 
\begin{talign*}
  \left\| \sum_{j=1}^T \tilde{w}_j k(\tilde{x}_j, \cdot) \right\|_{\Hk}^2 
    &= \sum_{j=1}^T \tilde{w}_j^2 k(\tilde{x}_j, \tilde{x}_j) + 2 \sum_{1 \leq j < k \leq T} \tilde{w}_j \tilde{w}_j k(\tilde{x}_j, \tilde{x}_k) \\
    &\leq \sum_{j=1}^T \tilde{w}_j^2 + 2 \sum_{1 \leq j < k \leq T} \tilde{w}_j \tilde{w}_k
\end{talign*}
hence 
\begin{talign}
   \E_{\P \sim \nplposterior} \left[\| \sum_{j=1}^T \tilde{w}_j k(\tilde{x}_j, \cdot) \|^2_{\Hk} \right]  &\leq \E_{w \sim \Dir} \left[\sum_{j=1}^T \tilde{w}_j^2 \right] + \E_{w \sim \Dir} \left[2 \sum_{1 \leq j < k \leq T} \tilde{w}_j \tilde{w}_k\right]. \label{eq:ineq2}
\end{talign}
Now from the moments of a Dirichlet distribution we have that for any $j,k \in \{1, \dots, T\}$:
\begin{talign*}
   \E_{w \sim \Dir}[\tilde{w}_j] = \frac{\alpha}{T(n+\alpha)}, \qquad  \E_{w \sim \Dir}[\tilde{w}_j^2] = \frac{\alpha^2 + T \alpha}{T^2(\alpha+n+1)(\alpha+n)} \qquad \text{ and } \qquad \E_{w \sim \Dir}[\tilde{w}_j \tilde{w}_k] &= \frac{\alpha^3 + \alpha^2 n}{T^2 (\alpha+n)^2 (\alpha + n + 1)}.
\end{talign*}
Substituting these in equation \eqref{eq:ineq2} we obtain
\begin{talign*}
    \E_{\P \sim \nplposterior} \left[\left\| \sum_{j=1}^T \tilde{w}_j k(\tilde{x}_j, \cdot) \right\|^2_{\Hk} \right] &\leq \frac{\alpha^2 + T \alpha}{T(\alpha+n+1)(\alpha+n)} + (T-1) \frac{\alpha^3 + \alpha^2 n}{T (\alpha+n)^2 (\alpha + n + 1)} \\
    &= \frac{T(\alpha+n)(\alpha+1)\alpha}{T(\alpha+n+1)(\alpha+n)} \\
    &= \frac{\alpha (1 + \alpha)}{(\alpha +n)(\alpha+n+1)}.
\end{talign*}
Hence, by Jensen's inequality
\begin{talign} \label{eq:ineq4}
\E_{\P \sim \nplposterior} \left[\left\| \sum_{j=1}^T \tilde{w}_j k(\tilde{x}_j, \cdot) \right\|_{\Hk} \right] \leq \sqrt{\frac{\alpha (1 + \alpha)}{(\alpha +n)(\alpha+n+1)}} 
\end{talign}
for some fixed $\alpha$. Therefore, by substituting equations \eqref{eq:ineq3} and \eqref{eq:ineq4} in equation \eqref{eq:ineq5} we obtain the required result:
\begin{talign*}
    \E_{x_{1:n} \simiid \truedist} \left[ \E_{\P \sim \nplposterior} \left[\MMD( \P, \Pn) \right]  \right] &\leq \E_{x_{1:n} \simiid \truedist} \left[  \left\|  \sum_{i=1}^n w_i k(x_i, \cdot) - \mu_{\Pn} \right\|_{\Hk} \right]  +  \E_{\P \sim \nplposterior} \left[ \left\| \sum_{j=1}^T \tilde{w}_j k(\tilde{x}_j, \cdot) \right\|_{\Hk} \right] \\
    &\leq \sqrt{\frac{2(n - 1) + \alpha(\alpha+1)}{(\alpha+n) (n + \alpha + 1)}} + \sqrt{\frac{\alpha (1 + \alpha)}{(\alpha +n)(\alpha+n+1)}}.
\end{talign*}
\end{proof}
We now proceed with the proof of the generalisation error in Theorem \ref{gen_error_bound}.
\begin{proof}
We have that for any $\theta \in \Theta$:
\begin{talign*}
    \MMD( \truedist, \P_{\theta^\ast(\P)}) &\leq \MMD(\P_{\theta^\ast(\P)}, \P) + \MMD(\truedist, \P) \\
    &\leq \MMD( \P_{\theta}, \P) +  \MMD(\truedist, \P)  \\
    &\leq \MMD( \P_{\theta}, \truedist) + 2 \MMD(\truedist, \P) \\
    &\leq \MMD( \P_{\theta}, \truedist) + 2 \MMD(\Pn, \truedist) + 2 \MMD(\P, \Pn)
\end{talign*}
For the first inequality above we used the triangle inequality and for the second inequality we used the fact that $\theta^\ast(\P) = \arg\inf_{\theta \in \Theta} \MMD(\P, \P_{\theta})$. The last two inequalities follow again by repeatedly applying the triangle inequality to $\MMD(\Ptheta, \P)$ and $\MMD(\truedist, \P)$. The above statements hold for any $\theta \in \Theta$ hence we have
\begin{talign*}
\MMD( \truedist, \P_{\theta^\ast(\P)}) \leq \inf_{\theta \in \Theta}  \MMD( \P_{\theta}, \truedist) + 2 \MMD(\Pn, \truedist) + 2 \MMD(\truedist, \Pn).
\end{talign*}
Taking double expectation on both sides we obtain 
\begin{talign*}
     \E_{x_{1:n} \simiid \truedist} \left[ \E_{\P \sim \nplposterior} \left[ \MMD(\truedist, \P_{\theta^\ast(\P)}) \right] \right] \leq &\inf_{\theta \in \Theta}  \MMD( \P_{\theta}, \truedist) + 2 \E_{x_{1:n} \simiid \truedist}\left[ \MMD(\Pn, \truedist)\right] \\
     &+ 2  \E_{x_{1:n} \simiid \truedist} \left[ \E_{\P \sim \nplposterior} \left[ \MMD(\P, \Pn) \right] \right].
\end{talign*}
The result then follows from Lemmas \ref{lemma7.1} and \ref{ptildephat}.
\end{proof}

The well-specified case is an immediate consequence of the Theorem:

\begin{corollary} \label{cor-gen-error-well}
Suppose that the model is well-specified, i.e. $\inf_{\theta \in \Theta}  \MMD( \P_{\theta}, \truedist) = 0$ then
\begin{talign*}
    \E_{x_{1:n} \simiid \truedist} \left[ \E_{\P \sim \nplposterior} \left[\MMD(\truedist, \P_{\theta^\ast(\P)}) \right] \right] 
    \leq \frac{2}{\sqrt{n}} + 2\sqrt{\frac{2(n - 1) + \alpha(\alpha+1)}{(\alpha+n) (n + \alpha + 1)}} + 2\sqrt{\frac{\alpha (1 + \alpha)}{(\alpha +n)(\alpha+n+1)}}.
\end{talign*}
\end{corollary}

\begin{proof}
The corollary follows directly from Theorem \ref{gen_error_bound} by setting $\inf_{\theta \in \Theta}  \MMD( \P_{\theta}, \truedist) = 0$.
\end{proof}

\subsection{Theorem \ref{post_consistency_miss}}
We can prove that the posterior consistency result indeed holds by adapting the arguments in \textcite{cherief2020mmd} (Theorem 2) from equation (\ref{Pi-case-consistency}) to the required statement of Theorem \ref{post_consistency_miss}.
\begin{proof}
Using Proposition \ref{gen_error_bound} and Markov's inequality we have:

\begin{talign*}
    &\E_{x_{1:n} \simiid \truedist} \left[ \nplposterior \left( \MMD( \P_{\theta^\ast(\P)}, \truedist) - \inf_{\theta \in \Theta} \MMD(\P_{\theta}, \truedist) > M_n \cdot n^{-1/2}\right)\right] \\
    &\leq \frac{\E_{x_{1:n} \simiid \truedist} \left[ \E_{\P \sim \nplposterior} \left[ \MMD( \P_{\theta^\ast(\P)}, \truedist) - C    \right]  \right]}{M_n \cdot n^{-1/2}} = \frac{\E_{x_{1:n} \simiid \truedist} \left[ \E_{\P \sim \nplposterior} \left[ \MMD( \P_{\theta^\ast(\P)}, \truedist)  \right]  \right] - C}{M_n \cdot n^{-1/2}} \\
    &\leq \frac{2 n^{-1/2}}{M_n \cdot n^{-1/2}} + \frac{2\sqrt{2n(n-1)+\alpha n(\alpha+1)}}{M_n\sqrt{(\alpha+n)(\alpha+n+1)}} + \frac{2 \sqrt{n \alpha(1 + \alpha)}}{M_n \sqrt{(\alpha + n)(\alpha + n + 1)}}  \overset{n\to\infty}{\longrightarrow} 0.
\end{talign*}

\end{proof}

Again, the well-specified case is a special case of Theorem \ref{post_consistency_miss} for $C = 0$.

\begin{corollary} \label{post_consistency_well}
Suppose we have a well-specified model, i.e. $\inf_{\theta \in \Theta} \MMD(\P_{\theta}, \truedist) = 0$ and let $\P = \sum_{i=1}^n w_i \delta_{x_i} + \sum_{k=1}^T \tilde{w}_k \delta_{\tilde{x}_k} \sim \nplposterior$ be a sample from the approximate posterior. Then for any $M_n \rightarrow + \infty$ such that $M_n \cdot n^{-1/2} \rightarrow 0$ as $n \rightarrow \infty$
\begin{talign*} 
    \nplposterior \biggl(  \P \in \PX :
    &\MMD(\truedist, \P_{\theta^\ast(\P)}) >  \frac{M_n}{n^{1/2}} \biggr) \overset{n\to\infty}{\longrightarrow} 0.
\end{talign*}
\end{corollary}
\begin{proof}
Using Theorem \ref{gen_error_bound} in the well specified case and Markov's inequality we have that 

\begin{talign*}
    \E_{x_{1:n} \simiid \truedist} \left[ \nplposterior \left( \MMD(\truedist, \P_{\theta^\ast(\P)}) > M_n \cdot n^{-1/2}\right)\right] &\leq \frac{\E_{x_{1:n} \simiid \truedist} \left[ \E_{\P \sim \nplposterior} \left[ \MMD(\truedist, \P_{\theta^\ast(\P)}) \right]  \right]}{M_n \cdot n^{-1/2}} \\
    &\leq \frac{2 n^{-1/2}}{M_n \cdot n^{-1/2}} + \frac{2\sqrt{2n(n-1)+\alpha n(\alpha+1)}}{M_n\sqrt{(\alpha+n)(\alpha+n+1)}} + \frac{2 \sqrt{n \alpha(1 + \alpha)}}{M_n \sqrt{(\alpha + n)(\alpha + n + 1)}}  \overset{n\to\infty}{\longrightarrow} 0.
\end{talign*}

\end{proof}

\subsection{Corollary \ref{cor-contam-mis}}
Before we present the generalisation bound we will use the following lemma from \textcite{cherief2019finite} which bounds the absolute difference of the MMD between the parametric model and each of the contaminated and non-contaminated probability measures. 

\begin{lemma}[\textcite{cherief2019finite} Lemma 3.3]\label{outliers_ineq}
For any $\theta \in \Theta$,
\begin{talign*}
    | \MMD(\P_{\theta}, \truedist) - \MMD(\P_{\theta}, \tilde{\P}) | \leq 2 \epsilon.
\end{talign*}
\end{lemma}
Using this lemma and the previously obtained generalisation bounds we can show the required statement of Corollary \ref{cor-contam-mis}.
\begin{proof}
From Lemma \ref{outliers_ineq} we have that for any $\theta \in \Theta$:
\begin{talign} 
    \MMD(\P_{\theta}, \tilde{\P}) \leq 2 \epsilon + \MMD(\P_{\theta}, \P^\ast_{0}) \label{ineq1} \\
    \MMD(\P_{\theta}, \truedist) \leq 2 \epsilon + \MMD(\P_{\theta}, \tilde{\P}) \label{ineq2}
\end{talign}
Hence, using equations \ref{ineq1}, \ref{ineq2} and Proposition  \ref{gen_error_bound} we have that
\begin{talign*}
  \E_{x_{1:n} \simiid \truedist} \left[ \E_{\P \sim \nplposterior} \left[ \MMD(\tilde{\P}, \P_{\theta^\ast(\P)}) \right] \right] &\leq 2 \epsilon + \E_{x_{1:n} \simiid \truedist} \left[ \E_{\P \sim \nplposterior} \left[ \MMD(\truedist, \P_{\theta^\ast(\P)}) \right] \right] \\
  &\leq 2 \epsilon + \inf_{\theta \in \Theta} \MMD(\truedist, \P_{\theta}) + \frac{2}{\sqrt{n}} + 2\sqrt{\frac{2(n - 1) + \alpha(\alpha+1)}{(\alpha+n) (n + \alpha + 1)}} + 2\sqrt{\frac{\alpha (1 + \alpha)}{(\alpha +n)(\alpha+n+1)}} \\ 
  &\leq 4 \epsilon + \inf_{\theta \in \Theta} \MMD(\tilde{\P}, \P_{\theta}) + \frac{2}{\sqrt{n}} + 2\sqrt{\frac{2(n - 1) + \alpha(\alpha+1)}{(\alpha+n) (n + \alpha + 1)}} + 2\sqrt{\frac{\alpha (1 + \alpha)}{(\alpha +n)(\alpha+n+1)}}.
\end{talign*}

\end{proof}

The well-specified case is an immediate consequence:
\begin{corollary}[Well-specified Case]\label{cor-contam-well} 
Suppose the model is well specified in terms of the target $\tilde{\P}$, i.e. $\exists \theta \in \Theta$ such that $\P_{\theta} = \tilde{\P}$. Then
\begin{talign*}
     \E_{x_{1:n} \simiid \truedist} \left[ \E_{\P \sim \nplposterior} \left[ \MMD( \tilde{\P}, \P_{\theta^\ast(\P)}) \right] \right] \leq 4 \epsilon + \frac{2}{\sqrt{n}} 
    + 2\sqrt{\frac{2(n - 1) + \alpha(\alpha+1)}{(\alpha+n) (n + \alpha + 1)}} + 2\sqrt{\frac{\alpha (1 + \alpha)}{(\alpha +n)(\alpha+n+1)}}.
\end{talign*}
\end{corollary}
\begin{proof}
The corollary follows from Corollary \ref{cor-contam-mis} noting that in the well-specified case $\inf_{\theta \in \Theta} \MMD(\P_{\theta}, \tilde{\P}) = 0$.
\end{proof}


\section{EXPERIMENTAL DETAILS} \label{app:experiments}

We now provide further experimental details of our method. We first explain how gradient-based numerical optimisation can be used to minimise the MMD objective in algorithm \ref{alg:MMD_posterior_bootstrap}. Next, we provide the experimental setup necessary for the reproduction of the experiments presented in the main text. We then give some additional details on the G-and-k and Toggle-Switch models. We finally provide results of all experiments over a number of independent runs.   

\subsection{Numerical Optimisation for the MMD Objective}
We first explain how the MMD objective in equation \eqref{eq:thetamap} can be approximated and minimized through gradient-based methods. For two probability measures $\P, \Q \in \mathcal{P}$ recall that the squared MMD is defined as 
\begin{talign*}
    \MMD^2 (\P, \Q) 
     = \int_{\X \times \X} &k(x,y) \P(dx) \P(dy)  
     - 2 \int_{\X \times \X} k(x,y) \P(dx) \Q(dy) \\
     &+ \int_{\X \times \X} k(x,y) \Q(dx) \Q(dy).
\end{talign*}
Since these integrals are usually intractable we can easily approximate this by sampling $x_{1:N} \simiid \P$ and $y_{1:m} \simiid \Q$ and using the U-statistic:
\begin{talign*}
   \widehat{\MMD}^2(\P, \Q) = \frac{1}{N(N-1)} &\sum_{i \neq i'}^N k(x_i, x_{i'}) - \frac{2}{Nm} \sum_{i=1}^N \sum_{j=1}^m k(x_i, y_j) \\
   &+ \frac{1}{m(m-1)} \sum_{j \neq j'} k(y_j, y_{j'}).
\end{talign*}
In our case we are interested in minimising the MMD between the parametric model $\Ptheta \in \PTheta$ defined through the simulator $G_{\theta}: \U \rightarrow \X$ and an approximated sample from the DP posterior $\P$. Hence, for samples $y_{1:N} \simiid \P$ and $u_{1:M} \simiid \U$ the approximated squared MMD is 
\begin{talign} 
\widehat{\MMD}^2(\P, \Ptheta) = \frac{1}{N(N-1)} &\sum_{i \neq i'}^N k(y_i, y_{i'}) - \frac{2}{NM} \sum_{j=1}^M \sum_{i=1}^N k(G_{\theta}(u_j), y_i) \nonumber \\
&+ \frac{1}{M(M-1)} \sum_{j \neq j'} k(G_{\theta}(u_j), G_{\theta}(u_{j'})). \label{eq:mmdapprox}
\end{talign}
The gradient of equation \eqref{eq:mmdapprox} can also be approximated using a U-statistic as follows: assuming that the generator is differentiable with respect to $\theta$ with gradient $\nabla_{\theta} G_{\theta}$ then the U-statistic of the gradient of the squared MMD is 
\begin{talign*}
   \hat{J}_{\theta}(u_{1:M}, y_{1:N}) := \nabla_{\theta} \widehat{\MMD}^2(\P, \P_{\theta}) = \frac{2}{M(M-1)} &\sum_{j \neq j'}^M \nabla_{\theta} G_{\theta}(u_j) \nabla_1 k(G_{\theta}(u_j), G_{\theta}(u_{j'})) \\
   &- \frac{2}{NM} \sum_{j=1}^M \sum_{i=1}^N \nabla_{\theta} G_{\theta}(u_j) \nabla_1 k(G_{\theta}(u_j), y_i)
\end{talign*}
where $\nabla_1$ denotes the partial derivative with respect to the first argument and by noting that the first term in \eqref{eq:mmdapprox} does not depend on $\theta$. This is an unbiased statistic in the sense that 
\begin{talign*}
\E_{u_{1:M} \simiid \U} \left[ \E_{y_{1:N \simiid \P}} [\nabla_{\theta} \widehat{\MMD}^2 (\P, \Ptheta)]\right] = \nabla_{\theta} \MMD^2(\P, \Ptheta).
\end{talign*}
Using this approximation we can use gradient-based methods to minimise the required objective. For example, one can use stochastic gradient descent as follows; for learning rate $\eta$, at each iteration $k = 1, 2, \dots$:
\begin{enumerate}
    \item Sample $u_{1:M} \simiid \U$ and $y_{1:N} \simiid \P$ 
    \item Set $\theta_k = \theta_{k-1} - \eta  \hat{J}_{\theta_{k-1}}(u_{1:M}, y_{1:N})$
\end{enumerate}

\subsection{Experimental Setup} \label{app:exp-details}

We provide details on the experimental setup required to produce figures \ref{fig-post-marg-gauss}, \ref{fig-post-marg-gnk} and \ref{fig-post-marg-togswitch}. For the WABC method we make use of the \texttt{winference R} package and the experimental setup provided in \textcite{bernton2019approximate}. For all experiments with our method we use \texttt{JAX} \parencite{bradbury2020jax} to parallelize the bootstrap sampling and perform the optimisation step. 

For all three experiments we use the Adam optimiser \parencite{kingma2014adam} for the minimization of the MMD at each bootstrap estimation. The learning rate is set to $\eta = 0.1$ for the Gaussian location and g-and-k models and to $\eta = 0.04$ for the toggle-switch model. To ensure convergence, we perform 2000 optimizations steps in the toggle-switch experiment and 1000 steps for the rest. 

The prior distributions set in the WABC method serve as an indication for the initialisation point of the optimization. In the Gaussian location model, a centered normal prior is imposed on the parameter, hence we initialize the optimization step at $\theta = 0$ at each bootstrap iteration. Similarly, the g-and-k model uses a Uniform prior on $[0,10]$ for all parameters so we initialise at the midpoint $(5,5,5,5)$. Finally, for the toggle-switch model we employ a random restart method as follows; 500 random values are sampled from the prior uniform distributions set in \textcite{bonassi2015sequential} and \textcite{bernton2019approximate}. The approximated MMD loss is computed for all of them and the three initial values with the smallest loss are chosen as initial points. The optimization is performed for all three starting points for each bootstrap iteration and the estimator with the smallest loss is retained. 

Finally, the length scale $l$ of the Gaussian kernel is set using the median heuristic $l = \sqrt{\text{median}_{1 \leq 1,j \leq n}(\| x_i - x_j \|^2_2)}$  suggested in \textcite{gretton2012kernel} and \textcite{dziugaite2015training} in the Gaussian location example while for the g-and-k model we set $l = 0.15$. For the toggle-switch model, we suggest an unweighted mixture of Gaussian kernels discussed in \textcite{DBLP:conf/iclr/SutherlandTSDRS17}, i.e. $k(x,x') = \sum_{i=1}^I \exp(-\|x-x'\|_2^2/(2 l_i^2))$ for a range of values $l_{i} \in \{1,10,20,40,80,100,130,200,400,800,1000\}$.

\subsection{The G-and-k Distribution}
We visualise the observed data from the G-and-k distribution used in the contaminated model example of section \ref{sec:experiments} for an increasing number of outliers with $\theta_0 = (3,1,1,-\log(2))$. We further provide the density obtained by generating 2000 samples from the model using the posterior mean $\theta = \frac{1}{B} \sum_{j=1}^B \theta^{(j)}$ for each method. The outliers of the contaminated models are visualised on the histograms of figure \ref{fig:gnkhists}. The densities of the fitted models show the sensitivity of the two methods to the two degrees of outliers; the densities fitted by the NPL-MMD remain similar and close to $\P_{\theta_0}$ whereas the densities fitted by WABC are affected by the outliers in the misspecifed cases.
 \begin{figure*}[t!]
\centering
\includegraphics[width=\textwidth]{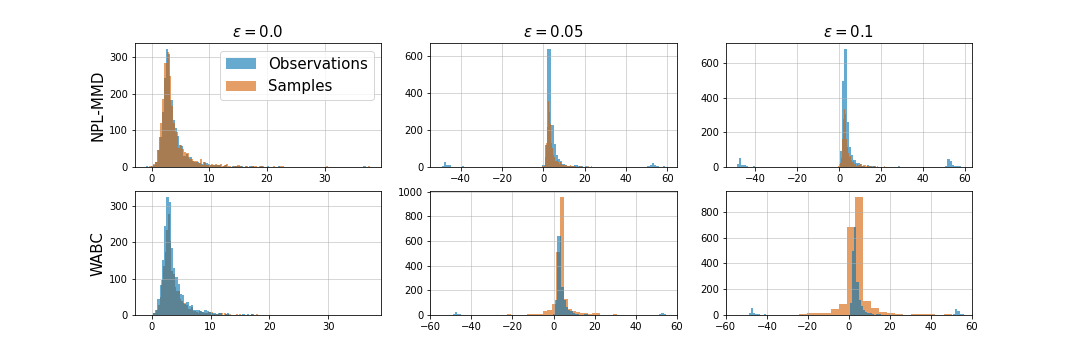}
\caption{Histogram of 2000 samples from the g-and-k distribution with $\theta_0 = (3,1,1,-\log(2))$ (blue) and 2000 samples from the fitted models using the posterior mean (orange) for the NPL-MMD method (above) and the WABC method (below).}
\label{fig:gnkhists}
\end{figure*}
 Figure \ref{fig:gnkboot} visualises the densities obtained at each bootstrap iteration of the NPL-MMD method. Here, each density corresponds to 2000 samples from the g-and-k distribution and parameter $\theta^{(j)}$ for $j = 1, \dots, B$. It is hence clear how each sample $\P^{(j)}$ from the DP posterior is mapped through $\theta^{*}(\cdot)$ to a value $\theta^{(j)}$ which gives rise to a different density. As we would expect, there is more variability in the densities for higher contamination levels. 
 \begin{figure*}[t!]
\centering
\includegraphics[width=\textwidth]{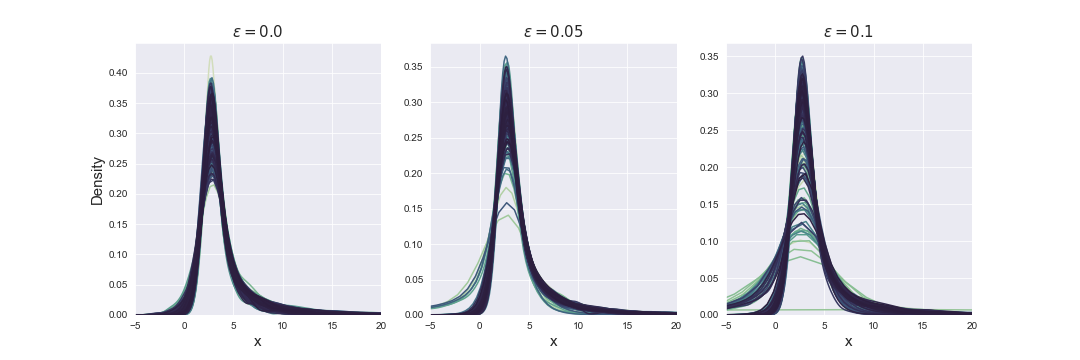}
\caption{Realisations of the G-and-k distribution with each $\theta^{(j)}$ where $j = 1, \dots, B$ obtained from the bootstrap iteration of the MMD posterior bootstrap algorithm for an increasing degree of contamination in the dataset.}
\label{fig:gnkboot}
\end{figure*}
Finally, we illustrate the rate obtained in Corollary \ref{cor-contam-mis} in the absence of outliers for the G-and-k distribution example. Since we take $\epsilon = \alpha = 0$, the Corollary implies that the expected MMD is bounded above by $\frac{2}{\sqrt{n}}$. To obtain an approximation of $\E_{x_{1:n} \simiid \truedist} \big[ \E_{\P \sim \nplposterior} \big[ \MMD^2 \big( \tilde{\P}, \P_{\theta^\ast(\P)}\big) \big] \big]$ we perform 10 runs of the algorithm and take the posterior mean $\hat{\theta} = \frac{1}{B} \sum_{i=1}^B \theta^{(i)}$ as the estimate of $\theta$. We then sample $15\times 10^3$ instances from $\truedist = \P_{\theta_{0}}$ and $\P_{\hat{\theta}}$ and estimate the squared MMD using the U-statistic in equation \eqref{eq:mmdapprox}. We repeat this experiment for ten values of sample size $n \in [250,4000]$ and plot $\sqrt{\E[\widehat{\MMD}^2(\P_{\hat{\theta}}, \truedist)]}$ against $\frac{2}{\sqrt{n}}$ in figure \ref{fig:sqrtnexp}. We observe that the estimate of the square root of the expected squared MMD is bounded above by the curve $\frac{2}{\sqrt{n}}$ and by Jensen's inequality it is implied that the estimate of the approximate MMD will also be bounded above since: 
$$
\E[\widehat{\MMD}(\P_{\hat{\theta}}, \truedist)] \leq \sqrt{\E[\widehat{\MMD}^2(\P_{\hat{\theta}}, \truedist)]} \leq \frac{2}{\sqrt{n}}.
$$

\subsection{The Toggle-Switch Model}\label{appendix:toggle_switch}

\subsubsection{Description of the Simulator}

For cell $i$ and unknown parameters $\theta = (\alpha_1,\alpha_2,\beta_1,\beta_2, \mu,\sigma, \gamma)^\top$, the simulator input is $u_i = (u_{i,1,1},u_{i,1,2},\ldots,u_{i,T,1},u_{i,T,2},u_{i,T+1,1})^\top \sim \text{Unif}([0,1]^{2T+1})$. The simulator $G_\theta$ is defined through:
\begin{talign*} 
G_\theta(u_i) & = \Phi^{-1}\Big(\Phi\Big(\frac{-(\mu +v_{i,T}) v_{i,T}^{\gamma}}{\mu \sigma}\Big)+u_{i,T+1,1} \Big( 1 - \Phi\Big(\frac{-(\mu +v_{i,T}) v_{i,T}^{\gamma}}{\mu \sigma}\Big) \Big) \Big) \frac{\mu \sigma}{v_{i,T}^{\gamma}} + (\mu +v_{i,T}) 
\end{talign*}
where for $t=1,\dots,T-1$, we have
\begin{talign*}
    \tilde{v}_{i,t+1} &= v_{i,t} + \frac{\alpha_1}{1+w_{i,t}^{\beta_1}} -(1 +0.03 v_{i,t}) \\
    \tilde{w}_{i,t+1} &= w_{i,t} + \frac{\alpha_2}{1+v_{i,t}^{\beta_2}} -(1 +0.03 w_{i,t}) \\
    v_{i,t+1} &= \tilde{v}_{i,t+1} + 0.5 \Phi^{-1}\big(\Phi(-2 \tilde{v}_{i,t+1}) + u_{i,t,1} (1-\Phi(-2\tilde{v}_{i,t+1}))\big) \\
    w_{i,t+1} &= \tilde{w}_{i,t+1} +0.5 \Phi^{-1}\big(\Phi(-2 \tilde{w}_{i,t+1})  + u_{i,t,2} (1-\Phi(-2\tilde{w}_{i,t+1}))\big) 
 \end{talign*} 
and $\Phi$ denotes the CDF of the standard Gaussian distribution. We use the initial conditions $v_{i,0} = 10$, $w_{i,0} = 10$. 
 \begin{figure*}[t!]
\centering
\includegraphics[width=0.5\textwidth]{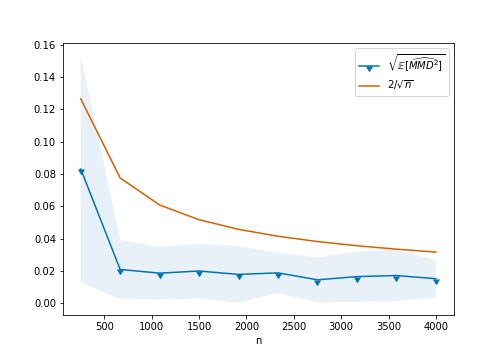}
\caption{Illustration of the upper bound of the generalisation error in Corollary \ref{cor-contam-mis} in the absence of outliers for the G-and-k distribution.}
\label{fig:sqrtnexp}
\end{figure*}

\subsection{Results Over Multiple Independent Runs} \label{app:table-results}
We provide results for a number of independent runs of the experiments in section \ref{sec:experiments}. For each run, a new dataset was generated and $B$ posterior samples were obtained for each parameter $\theta$. For each such sample, the mean estimator is recorded and the normalized mean squared error between the estimator and the true value of $\theta$ are presented with their standard deviations in tables \ref{table-gaussian}, \ref{table-gandk} and \ref{table-togswitch} for the Gaussian location, G-and-k and Toggle-Switch models respectively.  

\begin{table}[h!]
\caption{Experiment results for the Gaussian model over 10 runs} \label{table-gaussian}
\begin{center}
\begin{tabular}{c||c|c|c}
      \textbf{Method}  & \multicolumn{3}{c}{\textbf{NMSE (std)}} \\
      \hline \hline
     &  \textbf{$\epsilon=0$} & \textbf{$\epsilon=0.05$} &\textbf{$\epsilon=0.1$} \\ 
      \hline
      \textbf{NPL-MMD} & 0.0107 (0.00527)  & 0.00889 (0.00699) & 0.0113 (0.00538) \\
      \textbf{WABC} & 0.743 (0.0358) & 0.768 (0.0624) & 0.757 (0.0242) \\
      \textbf{NPL-WLL} & 0.00510 (0.00293) & 0.945 (0.0572) & 3.57 (0.0967) \\
      \textbf{NPL-WAS} & 0.00689 (0.00333) & 0.0189 (0.00627) & 0.0397 (0.0129) \\
      \textbf{MMD-ABC} & 0.750 (0.105) & 0.755 (0.0451)  &  0.760 (0.0411)
    \end{tabular}
    \end{center}
\end{table}
\begin{table}[!t]
\caption{Experiment results for the G-and-k distribution over 10 runs} \label{table-gandk}
\begin{center}
\begin{tabular}{c||c|c|c}
      \textbf{Method}  & \multicolumn{3}{c}{\textbf{NMSE (std)}} \\
      \hline \hline
     &  \textbf{$\epsilon=0$} & \textbf{$\epsilon=0.05$} &\textbf{$\epsilon=0.1$} \\ 
      \hline
      \textbf{NPL-MMD} & 0.00791 (0.00524) & 0.0128 (0.0117) & 0.0593 (0.0255) \\
      \textbf{WABC} & 0.00142 (0.000965) & 0.585 (0.0188)  & 0.532 (0.0109) \\
    \end{tabular}
    \end{center}
    \vspace{1cm}
\caption{Experiment results for the Toggle Switch model over 5 runs} \label{table-togswitch}
\begin{center}
\begin{tabular}{c||c}
      \textbf{Method}  & \textbf{NMSE (std)} \\
      \hline \hline
      \textbf{NPL-MMD} & 0.338 (0.277) \\
      \textbf{WABC} & 13.99 (13.8)  \\
    \end{tabular}
    \end{center}
\end{table}

\section{ADDITIONAL EXPERIMENTS} \label{app:additional-exps}
We provide some additional experiments for our method by considering a type of misspecification \texttt{other} than a contamination model, comparison with MMD-Bayes in \textcite{pacchiardi2021generalized} and exploring sensitivity to the DP prior (through $\alpha$ and  $\mathbb{F}$), Gaussian kernel length scale $l$ and truncation limit $T$. 
\subsection{Misspecified Gaussian Location Model}
So far in our empirical experiments we have considered the contamination model which is canonical for analysing misspecification in robust statistics, mostly because of theoretical convenience. While the model is simple, methods that are robust against contamination models usually fare well for more practically relevant alternatives, too.
Figure \ref{fig: mispecified_example} illustrates this on a new numerical example: We generate Cauchy-distributed data, but wrongly fit a Gaussian to it. The plot shows the posterior marginals for the location parameter, and vertically marks its true location. The MSEs over 20 runs was 0.0289 (NPL-MMD), 48.2 (NPL-WLL) and 28.6 (Standard Bayes).

\begin{figure}[t!]
\centering
	\includegraphics[width = 0.5\textwidth]{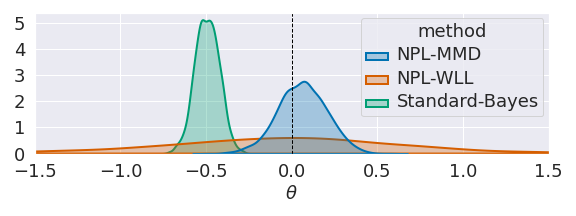}
	\caption{Marginal posterior distribution for mean of Normal location model with Cauchy data.}
	\label{fig: mispecified_example}
\end{figure}

\subsection{Comparison to MMD-Bayes for the Gaussian Location Model}	

We further compare our method to MMD-Bayes (\texttt{KernelScore}) in \textcite{pacchiardi2021generalized} on an $\varepsilon$-contaminated Gaussian location model with outliers at location $z$, where the weight is chosen using the grid search as in Section 4.2 of \textcite{pacchiardi2021generalized}. Figure \ref{fig: mmd_bayes} below, shows the marginal posterior distributions for Standard Bayes, Kernel Score and the MMD Posterior bootstrap methods for an increasing number of outliers and location parameter. 
\begin{figure}[t!]
\centering
        \includegraphics[width = 0.8 \textwidth]{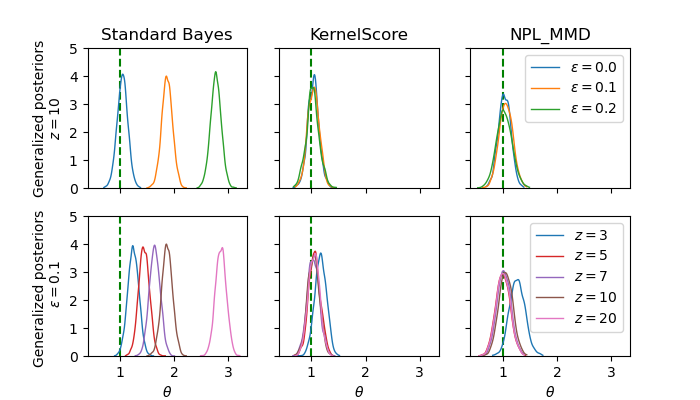}
	\caption{Contaminated Normal location model.}
	\label{fig: mmd_bayes}
\end{figure}

\subsection{Sensitivity to Hyperparameters}
In this section we empirically examine the sensitivity of the proposed method to several hyperparameters for the G-and-k distribution model. 

\subsubsection{Sensitivity to the DP Prior} \label{sec:alpha}
We first set $T = n$ in \eqref{eq-DP-post_approx} and $\mathbb{F}$ to be the Normal distribution with parameters equal to the mean and standard deviation of the observed data. Here we examine the effect of the prior in two ways; first by altering the hyperparameter $\alpha$ of the DP prior which characterizes how much certainty we impose on $\mathbb{F}$ and second by choosing an increasingly `worse' prior $\mathbb{F}$, since a higher proportion of outliers leads to a worse empirical estimates of the mean and standard deviation in the Normal prior. 

We generate $B = 2^9$ bootstrap samples $\theta_{1}, \dots, \theta_{B}$ for different values of $\alpha$ ranging in $[0.01, 300]$ and take the mean estimator $\hat{\theta} = \frac{1}{B} \sum_{b=1}^B \theta_b$ for each value of $\alpha$. In Figure \ref{fig:exp_hyp_alpha}, we plot (a) the normalised mean squared error for each estimator $\hat{\theta}$ for an increasing value of $\alpha$ and (b) samples from the generator with parameter $\hat{\theta}$. We observe that in the well-specified case, $\alpha$ has no significant effect in inference, however as we would expect, a larger value of $\alpha$, in combination with a worse prior centering measure increasingly affects the parameter inference. 

\begin{figure}[t!]
\centering
\subfloat[Normalised mean squared error of the obtained estimator for an increasing value of $\alpha$, corresponding to a higher confidence in the prior centering measure $\mathbb{F}$.]{
\includegraphics[width=0.4\textwidth]{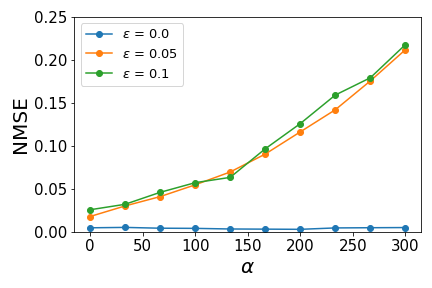}} \hfill
\subfloat[Generator samples obtained from each estimator $\hat{\theta}$. Lighter (resp. darker) curves correspond to a smaller (resp. larger) value of $\alpha$. The dotted line curve corresponds to the observed dataset in the well-specified case.]{
\includegraphics[width=\textwidth]{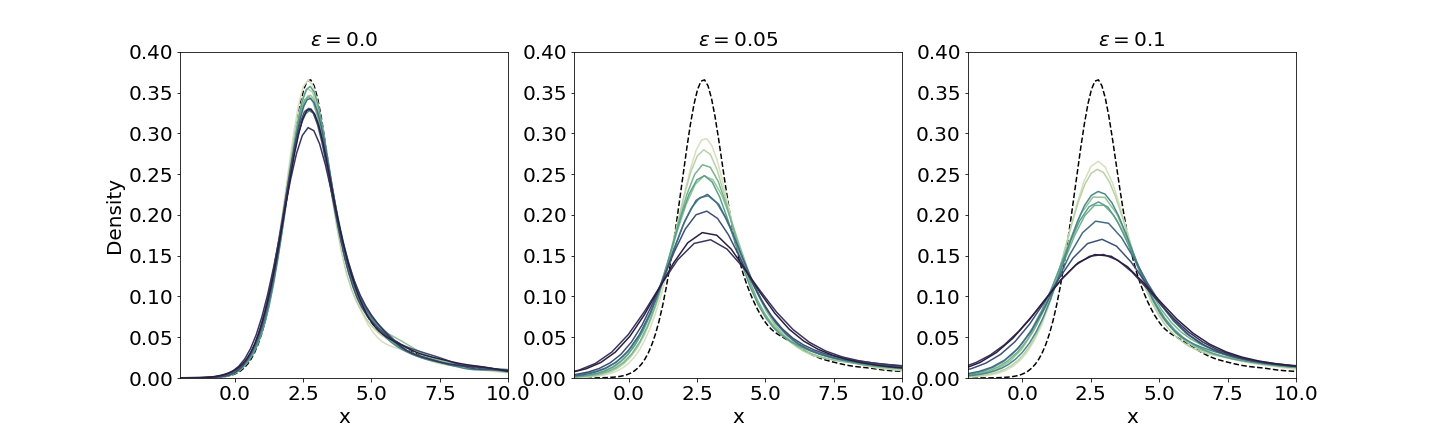}}
\caption{Sensitivity to DP prior.}
\label{fig:exp_hyp_alpha}
\end{figure}

\subsubsection{Sensitivity to $T$}
We further illustrate the effect of the truncation limit $T$ in the approximation of the DP posterior measure in \eqref{eq-DP-post_approx}. We fix $\alpha = 0.1$ and set $\mathbb{F}$ as in section \ref{sec:alpha}. We generate $B = 2^9$ bootstrap samples $\theta_{1}, \dots, \theta_{B}$ for different values of $T$ ranging in $[10,5000]$ and take the mean estimator $\hat{\theta} = \frac{1}{B} \sum_{b=1}^B \theta_b$ for each value of $T$. Figure \ref{fig:exp_hyp_T} shows (a) the normalised mean squared error for each estimator $\hat{\theta}$ for an increasing value of $T$ and (b) samples from the generator with parameter $\hat{\theta}$. We observe that in this example, the method is not significantly sensitive to the choice of $T$ for all values of $\epsilon$. 

\begin{figure}[t]
\centering
\subfloat[Normalised mean squared error of the obtained estimator for an increasing value of $T$.]{
\includegraphics[width=0.4\textwidth]{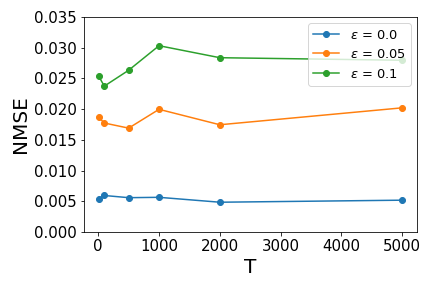}}
\hfill
\subfloat[Generator samples obtained from each estimator $\hat{\theta}$. Lighter (resp. darker) curves correspond to a smaller (resp. larger) value of $T$. The dotted line curve corresponds to the observed dataset in the well-specified case.]{\includegraphics[width=\textwidth]{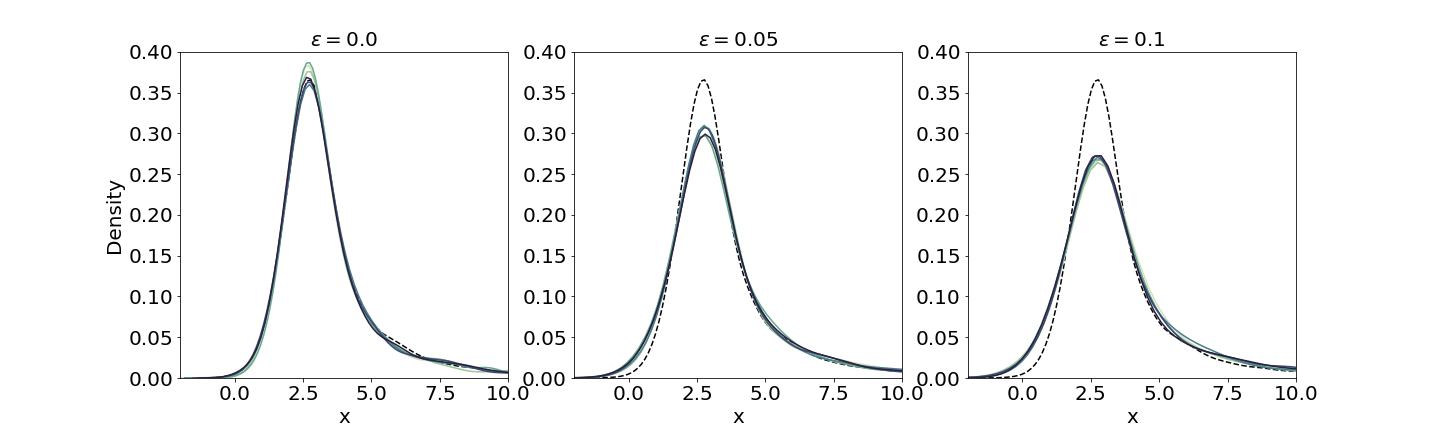}}
\caption{Sensitivity to parameter $T$, the truncation limit of the DP posterior.}
\label{fig:exp_hyp_T}
\end{figure}

\subsubsection{Sensitivity to the Hyperparameter of Gaussian Kernel}

We lastly investigate the effect of the length scale $l$ of the Gaussian kernel. We fix $\alpha = 0.1$, $T = n$ and set $\mathbb{F}$ as in \ref{sec:alpha}. We generate $B = 2^9$ bootstrap samples $\theta_{1}, \dots, \theta_{B}$ for different values of the length scale $l$ of the Gaussian kernel ranging in $[10^{-1},10^{2}]$ and take the mean estimator $\hat{\theta} = \frac{1}{B} \sum_{b=1}^B \theta_b$ for each value of $l$. Figure \ref{fig:exp_hyp_l} shows (a) the normalised mean squared error for each estimator $\hat{\theta}$ for an increasing value of $l$ in logarithmic scale and (b) samples from the generator with parameter $\hat{\theta}$. 

\begin{figure}[t!]
\centering
\subfloat[Normalised mean squared error of the obtained estimator for an increasing value of length scale (in logarithmic scale).]{\includegraphics[width=0.4\textwidth]{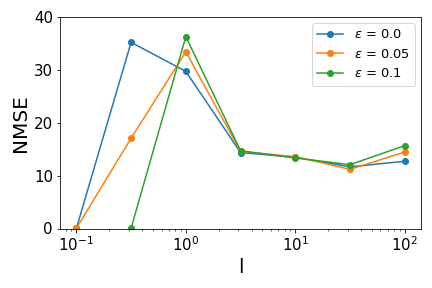}} \hfill
\subfloat[Generator samples obtained from each estimator $\hat{\theta}$. Lighter (resp. darker) curves correspond to a smaller (resp. larger) value of $l$. The dotted line curve corresponds to the observed dataset in the well-specified case.]{\includegraphics[width=\textwidth]{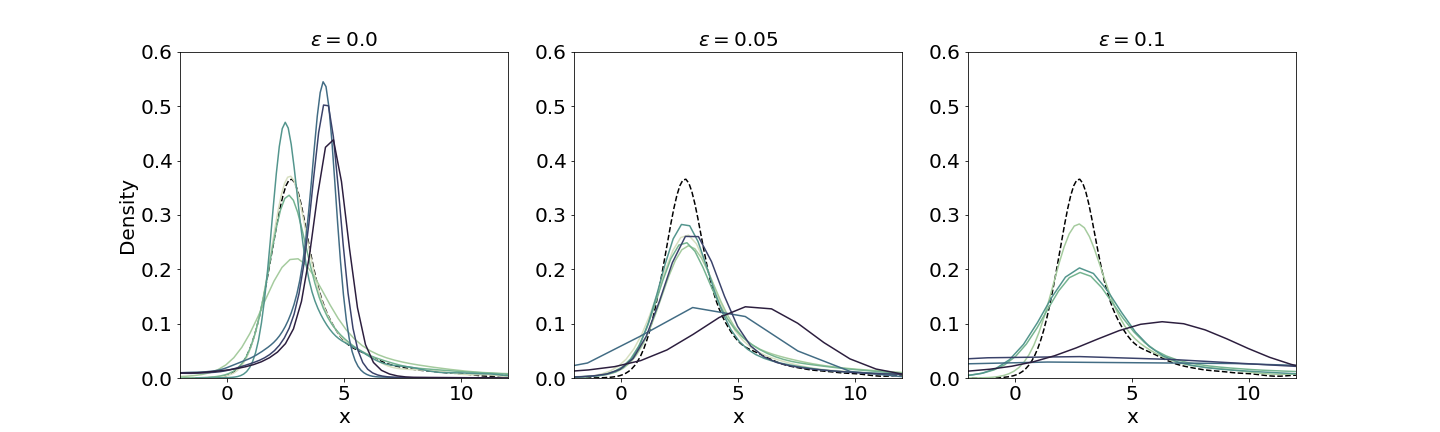}}
\caption{Sensitivity to the length scale of the Gaussian kernel.}
\label{fig:exp_hyp_l}
\end{figure}

To get some more intuition we plot the MMD loss as a function of $\theta_3 = g$ around a neighborhood of the true parameter value $\theta_0 = 1$. Since there is noise in our estimate of the MMD, it is possible that we get a global minimum for some value of $\theta_3$ far away from one. This is unlikely to happen for $l<1$ because there is a much bigger dip near $\theta_0$ as can be seen in figure \ref{fig:exp_hyp_l_loss}. Of course, this is just a projection of the loss landscape during optimisation, however it gives some intuition as to why a small choice of length scale in this model leads to better results. 
\begin{figure}[t!]
\centering
\includegraphics[width=0.5\textwidth]{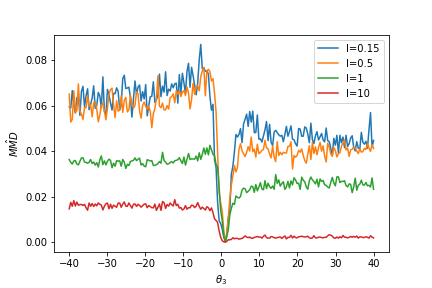}
\caption{Estimate of the MMD loss as a function of $\theta_3$ in the G-and-k distribution for different values of the length scale of the Gaussian kernel.}
\label{fig:exp_hyp_l_loss}
\end{figure}

\end{document}